\documentclass[11pt]{article}

\usepackage[margin=1in]{geometry}
\usepackage[T1]{fontenc}
\usepackage[utf8]{inputenc}
\usepackage{lmodern}
\usepackage{amsmath,amssymb,amsfonts,amsthm,mathtools,bm}
\usepackage{booktabs}
\usepackage{enumitem}
\usepackage{natbib}
\usepackage[colorlinks=true,linkcolor=blue,citecolor=blue,urlcolor=blue]{hyperref}
\usepackage{threeparttable}
\usepackage{siunitx}
\usepackage{pdflscape}
\usepackage{longtable}
\usepackage{adjustbox}
\usepackage[table]{xcolor}
\usepackage{array}
\usepackage{placeins}
\usepackage{longtable}

\newcommand{\gain}[1]{\textsuperscript{\scriptsize\textcolor{teal!70!black}{\,SA$\downarrow$#1\%}}}
\newcommand{\dash}{\textemdash}
\newcommand{\SArow}{\rowcolor{blue!7}}
\newcommand{\msegain}[2]{\makebox[3.0em][r]{#1}\gain{#2}}
\newcommand{\mseplain}[1]{\makebox[3.0em][r]{#1}}
\newcommand{\msebest}[1]{\makebox[3.0em][r]{\textbf{#1}}}

\newcolumntype{L}[1]{>{\raggedright\arraybackslash}p{#1}}
\newcolumntype{C}[1]{>{\centering\arraybackslash}p{#1}}
\newcolumntype{R}[1]{>{\raggedleft\arraybackslash}p{#1}}

\newtheorem{definition}{Definition}[section]
\newtheorem{assumption}{Assumption}[section]
\newtheorem{lemma}{Lemma}[section]
\newtheorem{proposition}{Proposition}[section]
\newtheorem{theorem}{Theorem}[section]
\newtheorem{corollary}{Corollary}[section]

\newcommand{\ind}{\mathbf{1}}

\newcommand{\V}{\mathbb{V}}
\newcommand{\Cov}{\operatorname{Cov}}
\newcommand{\Rbb}{\mathbb{R}}
\newcommand{\cS}{\mathcal S}
\newcommand{\cT}{\mathcal T}
\newcommand{\cX}{\mathcal X}
\newcommand{\cY}{\mathcal Y}
\newcommand{\Pn}{\mathbb{P}_n}
\newcommand{\Pnr}[1]{\mathbb{P}_{n,#1}}
\newcommand{\norm}[1]{\left\lVert #1 \right\rVert}
\newcommand{\abs}[1]{\left\lvert #1 \right\rvert}

\newcommand{\dto}{\rightsquigarrow}

\newcommand{\argmin}{\operatorname*{arg\,min}}

\title{Efficient Transported Distributional and Quantile Treatment Effects\\
with Surrogate-Assisted Missing Primary Outcomes}
\author{Pengyun Wang\\
Data Science Institute, The University of Chicago}

\date{\small Working manuscript. Comments welcome.\\This version: \today}

\begin{document}
\maketitle

\begin{abstract}
We study target-population distributional and quantile treatment effects when a source study observes treatment and post-treatment surrogates for all source units but observes a long-run primary outcome only for a validation subset, while the target population contributes only baseline covariates. The target estimands are transported counterfactual distribution functions $\psi_a(y)=P(Y^a\le y\mid R=0)$, their quantiles $q_a(\tau)$, and the quantile treatment effect $\Delta(\tau)=q_1(\tau)-q_0(\tau)$. The surrogate is not treated as a replacement endpoint and no Prentice-type surrogacy condition is imposed. Instead, the surrogate is used only to improve efficiency under missing-at-random primary-outcome sampling. We derive the nonparametric efficient influence function, which has three orthogonal components corresponding to target covariate sampling, the source surrogate process, and missing primary outcomes. This yields a closed-form cross-fitted one-step estimator after nuisance estimation. We establish identification, the canonical gradient, exact drift identities, ratio-level robustness, pointwise and uniform asymptotic linearity for transported CDFs, Bahadur representations for quantiles under explicit local inverse-map conditions, high-level multiplier-bootstrap simultaneous bands under explicit estimated-process and density conditions, and quantile-specific efficiency gains from observing surrogates. We also give lower-level nuisance-rate verification for a deliberately restricted class of analyzable bounded finite-dimensional or finite-rank implementations based on sieve ridge regression, ridge logistic regression, calibrated density-ratio estimation, finite-rank kernel ridge regression, and isotonic projection under explicit grid, eigenvalue, source, and entropy conditions.
\end{abstract}


\section{Introduction}

Long-run primary outcomes are often expensive, delayed, or observable only after linkage to administrative records, whereas short-run surrogate outcomes can be measured broadly. In a job-training study, for instance, treatment assignment, baseline covariates, program participation, and short-run earnings may be available for all source-study participants, while long-run earnings are observed only for a validation subset. Separately, a policymaker may have a large target-population covariate sample and wish to know how deploying the intervention would affect the long-run outcome distribution in that target population. Average treatment effects may be insufficient: an intervention can improve the lower tail while leaving the median or upper tail nearly unchanged.

We consider independent observations
\[
O=(R,X,R(A,S,M,MY)),
\]
where $R=1$ indexes a source study and $R=0$ indexes the target covariate sample. In the source, $A\in\{0,1\}$ is treatment, $S$ is a post-treatment surrogate observed for all source units, $Y$ is the long-run primary outcome, and $M$ indicates whether $Y$ is observed. In the target, only $X$ is observed. The target quantities are
\[
\psi_a(y)=P(Y^a\le y\mid R=0),\qquad
q_a(\tau)=\inf\{y:\psi_a(y)\ge \tau\},\qquad
\Delta(\tau)=q_1(\tau)-q_0(\tau).
\]

This paper has four theoretical contributions. First, we identify the transported treatment-specific CDF by combining conditional transportability of the treatment-specific surrogate-primary-outcome process with missing-at-random validation sampling of the primary outcome. Second, we derive the canonical gradient of this observed-data functional. The gradient has three orthogonal residual components: a target-covariate component, a source-surrogate-process component, and a primary-outcome-validation component. Third, the resulting one-step estimator satisfies an exact drift identity, yielding ratio-level robustness and product-rate conditions for nuisance estimation. Fourth, we show how the distributional theory induces pointwise and uniform quantile inference and a quantile-specific efficiency-gain formula for observing surrogates. The surrogate is used as auxiliary information; it is not treated as a replacement endpoint and no Prentice-type surrogacy condition is imposed.

The main practical object is the closed-form one-step estimator in Definition \ref{def:estimator}. Given nuisance estimates, it is only a sample average. A finite-grid monotone projection and quantile inversion complete the procedure. The proposed implementation is deliberately analyzable: the theory verifies a specific class of bounded finite-dimensional or finite-rank nuisance learners, including finite-rank kernel ridge or ridge-linear regressions, ridge logistic regressions, calibrated density-ratio estimators such as entropy balancing under an exponential-tilt model, and isotonic projection. These subproblems are normal-equation based or convex, but the lower-level verification is not a claim about arbitrary flexible convex learners.

The paper is organized as follows. Section \ref{sec:related} reviews the related literature. Section \ref{sec:setup} fixes notation, the observed-data structure, and the estimands. Section \ref{sec:theory} gives identification, efficient influence functions, one-step estimation, uniform inference, low-level nuisance verification, and efficiency-gain results. Section \ref{sec:experiments} outlines simulation designs. Section \ref{sec:realdata}  provide an empirical illustration using the ACTG 175 HIV clinical trial, where 20-week CD4/CD8 measurements serve as early surrogate biomarkers and 96-week CD4 count is the incompletely observed primary outcome \citep{hammer1996trial,actg175data}. Proofs are collected in the Appendix.

\section{Related Literature}\label{sec:related}

\paragraph{Surrogates, validation samples, and long-run outcomes.}
Classical surrogate-endpoint work asks when short-run variables can replace the primary outcome, often through conditions such as statistical surrogacy, principal surrogacy, or candidate-principal-surrogate criteria \citep{Prentice1989,BuyseMolenberghs1998,FrangakisRubin2002,GilbertHudgens2008,VanderWeele2013}. Modern long-term causal inference also uses surrogate or proxy information to connect short-run experiments with long-run outcomes \citep{AtheyChettyImbensKang2019,ChenRitzwoller2023,ImbensKallusMaoWang2022}. In contrast, our estimand remains the long-run primary outcome distribution. The surrogate is not an endpoint replacement and is not required to mediate the treatment effect. Our closest precursor is \citet{KallusMao2024Surrogates}, who study semiparametric efficiency for average treatment effects with limited primary outcomes and abundant surrogates. We extend this line from first moments to transported counterfactual CDFs and quantiles, and we characterize quantile-specific efficiency gains.

\paragraph{Quantile, distributional, and missing-data causal inference.}
Classical quantile regression and distributional treatment-effect analysis provide the statistical background for QTE inference \citep{KoenkerBassett1978,AbadieAngristImbens2002,ChernozhukovHansen2005,ChernozhukovFernandezValGalichon2010,ChernozhukovFernandezValMelly2013}. Efficient quantile treatment effect estimation under unconfoundedness was developed by \citet{Firpo2007}; extensions include unconditional QTE under endogeneity \citep{FrolichMelly2013}, efficient quantile estimation in missing-data models \citep{Diaz2017}, and localized debiased machine learning for QTEs and related estimands \citep{KallusMaoUehara2024LDML}. Distributional causal inference beyond means is also closely related to semiparametric counterfactual density estimation \citep{KennedyBalakrishnanWasserman2021}. Our paper uses a full-CDF and growing-grid route rather than the LDML localization route: the contribution is not localization, but the canonical-gradient and inference theory for a two-population surrogate-assisted validation-outcome model.

\paragraph{Semi-supervised inference and data fusion.}
Semi-supervised inference has developed efficiency gains from unlabeled covariate samples in regression and causal-effect problems \citep{ChakraborttyCai2018,ChakraborttyDaiTchetgen2024,ChakraborttyDaiCarroll2022}. Data-fusion methods for QTEs combine validation and auxiliary data under different missingness or confounding structures \citep{ZhangZhu2023FQTE}. These works are important neighboring problems, but our target sample is a target population covariate sample rather than merely a larger unlabeled sample from the same population, and our source sample contains post-treatment surrogates and validation-sampled primary outcomes. The target is the transported distribution $P(Y^a\le y\mid R=0)$, not a source-population QTE.

\paragraph{Transportability and target-population inference.}
Generalizability and transportability methods adjust for differences between the study population and a target population, mostly for average or other first-moment causal measures \citep{ColeStuart2010,DahabrehRobertson2019,TaoFu2019,LeeYangWang2022,BoughdiriBerenfeldJosseScornet2025}. The weighting ideas used in this literature are closely connected to propensity-score and semiparametric efficiency results for treatment-effect estimation \citep{RosenbaumRubin1983,Hahn1998,HiranoImbensRidder2003}. Our work is different in two respects. First, the target is a counterfactual distribution and its quantile function, so inference must pass through an inverse-CDF map. Second, source outcomes are not fully observed and post-treatment surrogates enter the efficiency calculation. The density-ratio and target-covariate component in our canonical gradient are therefore coupled with surrogate and validation-outcome residuals.

\paragraph{Semiparametric efficiency and orthogonal estimation.}
The technical tools build on semiparametric efficiency theory for missing and coarsened data \citep{RobinsRotnitzkyZhao1994,RobinsRotnitzky1995,Newey1994,BangRobins2005,vdLaanRobins2003,Tsiatis2006}, double/debiased machine learning and cross-fitting \citep{BelloniChernozhukovFernandezValHansen2017,ChernozhukovDML2018}, and empirical-process theory for both Donsker and high-dimensional suprema \citep{vdVW1996,ChernozhukovChetverikovKato2014}. Our exact drift identity is the analogue of an orthogonality calculation for the transported CDF functional; it gives product-rate sufficient conditions for nuisance learning and explains the ratio-level robustness of the estimator.

\section{Notation, Observed Data, and Estimands}\label{sec:setup}

All random variables are defined on a common probability space. The observed-data law is $P$ and expectation under $P$ is $E$. For $r\in\{0,1\}$, let
\[
\pi_r=P(R=r),\qquad E_r\{h(X)\}=E\{h(X)\mid R=r\},
\]
let $P_r$ denote the conditional law given $R=r$, and let $P_{X,r}$ denote the law of $X$ conditional on $R=r$. For any integrable function $f$, write $P_rf=E(f\mid R=r)$ and $Pf=E f$. For source conditional expectations, all conditioning events are implicitly intersected with $R=1$ unless $R$ is explicitly displayed. For a sample $O_1,\ldots,O_n$, write $n_r=\sum_{i=1}^n\ind(R_i=r)$ and define
\[
\Pnr{r} f=\frac{1}{n_r}\sum_{i:R_i=r}f(O_i),\qquad r=0,1.
\]
The unconditional empirical mean is $\Pn f=n^{-1}\sum_{i=1}^n f(O_i)$. Norms $\norm{f}_{Q,2}$ are $L_2(Q)$ norms. We write $\norm{h}_{1a,2}$ for the $L_2$ norm under the law of $(X,S)$ conditional on $(R=1,A=a)$.

\begin{definition}[Observed data]\label{def:obsdata}
The observed sigma-field is generated by
\[
O=(R,X,R(A,S,M,MY)),
\]
where $R\in\{0,1\}$, $X\in\cX$, $A\in\{0,1\}$, $S\in\cS$, $M\in\{0,1\}$, and $Y\in\Rbb$. Equivalently, when $R=0$ only $X$ is observed; when $R=1$, $A,S,M$ are observed and $Y$ is observed if and only if $M=1$. Symbols $A,S,M,Y$ may be assigned arbitrary placeholder values when they are unobserved; all formulas multiply them by the relevant observation indicators and hence depend only on the observed sigma-field.
\end{definition}

\begin{definition}[Potential outcomes and target estimands]\label{def:estimands}
Following the potential-outcomes convention for defining causal effects as comparisons across interventions on the same units \citep{Rubin1978,Rubin2005}, for $a\in\{0,1\}$, let $(S^a,Y^a)$ be the potential surrogate and primary outcome under treatment $a$. The target counterfactual distribution function, target counterfactual quantile, and target QTE are
\[
\psi_a(y)=P(Y^a\le y\mid R=0),\qquad
q_a(\tau)=\inf\{y:\psi_a(y)\ge\tau\},\qquad
\Delta(\tau)=q_1(\tau)-q_0(\tau).
\]
The distributional treatment effect is $\psi_1(y)-\psi_0(y)$.
\end{definition}

\begin{definition}[Nuisance functions]\label{def:nuisance}
For $a\in\{0,1\}$ and $y\in\Rbb$, define
\begin{align*}
m_a(y,x,s)&=P(Y\le y\mid R=1,A=a,X=x,S=s,M=1),\\
g_a(y,x)&=E\{m_a(y,X,S)\mid R=1,A=a,X=x\},\\
e_a(x)&=P(A=a\mid R=1,X=x),\\
\rho_a(x,s)&=P(M=1\mid R=1,A=a,X=x,S=s).
\end{align*}
If $P_{X,0}\ll P_{X,1}$, define
\[
\omega(x)=\frac{dP_{X,0}}{dP_{X,1}}(x).
\]
Let $Z_y=\ind(Y\le y)$. The nuisance collection is $\eta=(m,g,e,\rho,\omega)$.
\end{definition}

\section{Theory: Identification, Efficiency, and Inference}\label{sec:theory}

The theory proceeds in four steps.  First, the causal and sampling assumptions reduce the target counterfactual CDF to a regression functional of the observed-data law.  Second, the canonical gradient of this functional is derived under the observed-data model.  Third, the same gradient motivates a cross-fitted one-step estimator and gives the asymptotic expansion needed for CDF and quantile inference.  Fourth, the high-level nuisance conditions are connected to specific bounded convex or finite-rank implementations, and the role of observing the surrogate is quantified through an efficiency comparison.

\subsection{Identification Assumptions}\label{sec:identification}

The first issue is identification of $\psi_a(y)$ from the two-sample observed-data law.  The assumptions below separate the required restrictions by their roles.  Assumption \ref{ass:causal} links the source factual data to potential outcomes.  Assumption \ref{ass:mar} justifies using the validation subset to learn the conditional law of $Y$ given $(X,S,A)$.  Assumption \ref{ass:transport} moves the treatment-specific law of $(S^a,Y^a)$ from the source population to the target population conditional on $X$.  Assumption \ref{ass:positivity} ensures that all conditional regressions and reweighting operations appearing later are well defined.

\begin{assumption}[Consistency and source exchangeability]\label{ass:causal}
For source units, $S=S^A$ and $Y=Y^A$. For each $a\in\{0,1\}$,
\[
(S^a,Y^a)\perp A\mid X,R=1.
\]
This condition is automatic in a randomized source trial, possibly after conditioning on design strata included in $X$.
\end{assumption}

Assumption \ref{ass:causal} is only a source-study condition.  It does not require the source and target covariate distributions to be the same, and it does not impose any surrogate-endpoint condition.  Its role is to let the treatment-specific source arm with $A=a$ stand in for the corresponding potential-outcome law among source units with the same $X$, using the same conditional-exchangeability logic underlying propensity-score adjustment \citep{RosenbaumRubin1983}.

\begin{assumption}[Primary-outcome missing at random]\label{ass:mar}
For each $a\in\{0,1\}$, in the source factual arm,
\[
M\perp Y^a\mid R=1,A=a,X,S^a.
\]
Equivalently, by consistency, the conditional law of the factual primary outcome $Y$ among source units with $A=a$ does not depend on $M$ after conditioning on $(X,S)$:
\[
Y\perp M\mid R=1,A=a,X,S.
\]
Thus primary-outcome labeling may depend on treatment, baseline covariates, and the post-treatment surrogate, but not on the unobserved primary outcome after conditioning on these variables.
\end{assumption}

Given Assumption \ref{ass:causal}, the remaining source-study complication is that $Y$ is only observed when $M=1$.  Assumption \ref{ass:mar} states that, after conditioning on $(X,S)$ within the factual arm $A=a$, the validation subset has the same conditional primary-outcome law as the full source arm.  This is why the regression $m_a(y,x,s)$ in Definition \ref{def:nuisance} is defined using $M=1$ but can still represent the full source-arm conditional distribution.

\begin{assumption}[Conditional transportability]\label{ass:transport}
For each $a\in\{0,1\}$,
\[
(S^a,Y^a)\mid X,R=0\;\overset{d}{=}\;(S^a,Y^a)\mid X,R=1.
\]
\end{assumption}

Assumption \ref{ass:transport} is the bridge from the source study to the target population.  It allows $P_{X,0}$ and $P_{X,1}$ to differ; that difference is handled by the density ratio $\omega$.  The restriction is instead on the conditional treatment-specific law given $X$, which is the object transported across populations.

\begin{assumption}[Positivity and boundedness]\label{ass:positivity}
There exists $\epsilon>0$ such that $\pi_0,\pi_1\ge\epsilon$, $P_{X,0}\ll P_{X,1}$, $\omega(X)\le \epsilon^{-1}$ $P_{X,1}$-almost surely,
\[
e_a(X)\ge\epsilon\quad P(\cdot\mid R=1)\text{-almost surely},
\]
and
\[
\rho_a(X,S)\ge\epsilon\quad P(\cdot\mid R=1,A=a)\text{-almost surely}.
\]
The conditional distribution functions $m_a$ and $g_a$ take values in $[0,1]$.
\end{assumption}

With these conditions in place, the target CDF can be expressed without unobserved potential outcomes.  The proof follows a single conditioning chain: start from the target law of $X$, invoke conditional transportability, use source exchangeability to condition on $A=a$, use missing-at-random sampling to replace the unobserved primary-outcome law by the validated regression $m_a$, and finally integrate out $S$ to obtain $g_a$.

\begin{proposition}[Identification]\label{prop:identification}
Under Assumptions \ref{ass:causal}--\ref{ass:positivity}, for each $a\in\{0,1\}$ and $y\in\Rbb$,
\[
\psi_a(y)=E_0\{g_a(y,X)\}.
\]
Equivalently,
\[
\psi_a(y)=E\left\{\frac{\ind(R=0)}{\pi_0}g_a(y,X)\right\}=E_1\{\omega(X)g_a(y,X)\}.
\]
\end{proposition}

Proposition \ref{prop:identification} is the point at which the causal problem becomes an observed-data functional.  The representation $E_0\{g_a(y,X)\}$ emphasizes that the final averaging is over the target covariate distribution, while the equivalent representation $E_1\{\omega(X)g_a(y,X)\}$ shows how the same quantity can be written as a source-population weighted expectation.  The next subsection studies the efficient influence function for this identified functional.

\subsection{Efficient Influence Functions and Drift Identities}\label{sec:eif}

The efficient influence function is written in a form that mirrors the three pieces of information used for identification: target covariates, source surrogates, and validated primary outcomes.  This decomposition is also useful for estimation because each residual term will later become a correction term in the one-step estimator.

\begin{definition}[Efficient influence function candidate]\label{def:eif}
For fixed $a\in\{0,1\}$ and $y\in\Rbb$, define
\begin{align}
\phi_a^y(O)
&=\frac{\ind(R=0)}{\pi_0}\{g_a(y,X)-\psi_a(y)\} \notag\\
&\quad +\frac{\ind(R=1)\omega(X)\ind(A=a)}{\pi_1 e_a(X)}\{m_a(y,X,S)-g_a(y,X)\} \notag\\
&\quad +\frac{\ind(R=1)\omega(X)\ind(A=a)M}{\pi_1 e_a(X)\rho_a(X,S)}\{Z_y-m_a(y,X,S)\}.
\label{eq:eif}
\end{align}
\end{definition}

The first line of $\phi_a^y$ is the fluctuation from sampling target covariates.  The second line is the source surrogate-process residual $m_a-g_a$ after reweighting the source arm to the target covariate law.  The third line is the validation-outcome residual $Z_y-m_a$ among units with observed primary outcomes.  Each line is centered with respect to the conditioning information of the corresponding data-generating component, which is why the decomposition is orthogonal in the nonparametric observed-data model.

\begin{theorem}[Canonical gradient for the transported observed-data CDF]\label{thm:eif}
Let
\[
\Psi_a^y(P)=E_{P}\{g_{a,P}(y,X)\mid R=0\},
\]
where $g_{a,P}$ is the conditional surrogate-integrated regression induced by the observed-data law $P$ in Definition \ref{def:nuisance}. Under the nonparametric observed-data model induced by the factorization
\begin{align*}
p(o)&=p_R(0)p_{X\mid R}(x\mid0),\qquad &&r=0,\\
p(o)&=p_R(1)p_{X\mid R}(x\mid1)p_{A\mid X,R}(a\mid x,1)p_{S\mid A,X,R}(s\mid a,x,1)\\
&\quad\times p_{M\mid A,X,S,R}(m\mid a,x,s,1)\{p_{Y\mid M,A,X,S,R}(y\mid1,a,x,s,1)\}^{m},\qquad &&r=1,
\end{align*}
and the positivity part of Assumption \ref{ass:positivity}, the function $\phi_a^y$ in \eqref{eq:eif} is the canonical gradient of the observed-data functional $\Psi_a^y(P)$. Under Assumptions \ref{ass:causal}--\ref{ass:positivity}, Proposition \ref{prop:identification} gives $\Psi_a^y(P)=\psi_a(y)$, so the same function is the efficient influence function for the causal transported CDF. Consequently, the semiparametric efficiency bound for estimating $\psi_a(y)$ is
\[
V_a(y)=E\{[\phi_a^y(O)]^2\}.
\]
The same canonical gradient and bound obtain in submodels in which $e_a$ or $\rho_a$ is known by design.
\end{theorem}

Theorem \ref{thm:eif} confirms that the candidate in Definition \ref{def:eif} is not merely an estimating equation but the canonical gradient.  The calculation follows the standard pathwise-differentiability and tangent-space approach to semiparametric efficiency \citep{Newey1994,RobinsRotnitzky1995,Tsiatis2006}.  Thus any regular estimator of $\psi_a(y)$ must have asymptotic variance at least $V_a(y)$, and a one-step estimator with this leading influence function is semiparametrically efficient at fixed $y$.  The statement that the gradient is unchanged when $e_a$ or $\rho_a$ is known reflects the same conditional centering: the displayed gradient is already orthogonal to the corresponding treatment and validation-sampling score directions.

\begin{corollary}[Mean-outcome analogue and relation to surrogate ATE]\label{cor:mean_special_case}
Define
\[
\mu_a=E_0\{\mu_a^S(X)\},\qquad
\mu_a^S(x)=E\{\widetilde\mu_a(X,S)\mid R=1,A=a,X=x\},
\]
where $\widetilde\mu_a(x,s)=E(Y\mid R=1,A=a,X=x,S=s,M=1)$. Assume in addition that the displayed influence function below is square-integrable; equivalently, the target, surrogate-process, and validation-outcome residual terms below have finite second moments under the corresponding weighted laws. Under Assumptions \ref{ass:causal}--\ref{ass:positivity} and this square-integrability condition, the efficient influence function for the transported mean $\mu_a$ in the present two-sample transport model is
\begin{align*}
\phi_{\mu_a}(O)
&=\frac{\ind(R=0)}{\pi_0}\{\mu_a^S(X)-\mu_a\}
+\frac{\ind(R=1)\omega(X)\ind(A=a)}{\pi_1e_a(X)}\{\widetilde\mu_a(X,S)-\mu_a^S(X)\}\\
&\quad+\frac{\ind(R=1)\omega(X)\ind(A=a)M}{\pi_1e_a(X)\rho_a(X,S)}\{Y-\widetilde\mu_a(X,S)\}.
\end{align*}
Thus the efficient influence function for the transported ATE $\mu_1-\mu_0$ is $\phi_{\mu_1}-\phi_{\mu_0}$. If $P_{X,0}=P_{X,1}$, then $\omega=1$ and the residual decomposition has the same baseline, surrogate-process, and validation-outcome structure as the surrogate-assisted ATE canonical gradient of \citet{KallusMao2024Surrogates}. It is not literally the same gradient unless the two-sample transport design is collapsed to the corresponding single-population observation model; under that additional collapse, with $M$ identified with the primary-outcome labeling indicator, the displayed armwise gradient reduces exactly to the reference armwise surrogate-assisted ATE gradient.
\end{corollary}

Corollary \ref{cor:mean_special_case} is a consistency check against the mean-outcome literature.  Replacing $Z_y$ by $Y$ preserves the same target, surrogate, and validation residual structure.  The CDF theory, however, is more than a pointwise mean result because quantile inference requires passing the whole distributional estimator through an inverse-CDF map.

\begin{assumption}[Quantile regularity]\label{ass:quantile}
For each $a\in\{0,1\}$, $\psi_a$ is continuously differentiable on an open neighborhood of $\{q_a(\tau):\tau\in\cT\}$, where $\cT\subset(0,1)$ is compact and the relevant quantiles lie away from boundary points of that neighborhood. Let $f_a(y)=\partial_y\psi_a(y)$. There exists $c_f>0$ such that
\[
\inf_{a\in\{0,1\}}\inf_{\tau\in\cT} f_a(q_a(\tau))\ge c_f.
\]
\end{assumption}

Assumption \ref{ass:quantile} supplies the local invertibility needed to convert CDF perturbations into quantile perturbations.  The lower bound on $f_a(q_a(\tau))$ rules out flat portions of the target CDF on the quantile region, so first-order errors in $\psi_a$ translate stably into first-order errors in $q_a$.

\begin{corollary}[Efficient influence functions for quantiles and QTEs]\label{cor:eif_quantile}
Under Assumptions \ref{ass:causal}--\ref{ass:quantile}, the efficient influence function for $q_a(\tau)$ is
\[
\phi_{q_a,\tau}(O)=-\frac{\phi_a^{q_a(\tau)}(O)}{f_a(q_a(\tau))}.
\]
The efficient influence function for $\Delta(\tau)$ is
\[
\phi_{\Delta,\tau}(O)=
-\frac{\phi_1^{q_1(\tau)}(O)}{f_1(q_1(\tau))}
+\frac{\phi_0^{q_0(\tau)}(O)}{f_0(q_0(\tau))}.
\]
The corresponding efficiency bound is $E\{\phi_{\Delta,\tau}^2(O)\}$.
\end{corollary}

The efficient influence function for the QTE is obtained by evaluating the arm-specific CDF influence functions at the arm-specific target quantiles and scaling by the corresponding target densities.  This is why later QTE inference requires both a good CDF expansion and control of the local inverse map.  The next step is to understand what happens when the nuisance functions in the CDF signal are replaced by estimates.

\begin{definition}[Candidate signal and drift notation]\label{def:signal}
Let $\bar\eta=(\bar m,\bar g,\bar e,\bar\rho,\bar\omega)$ be a candidate nuisance collection satisfying the same positivity bounds as Assumption \ref{ass:positivity}. Define
\begin{align*}
\Psi_a^y(\bar\eta)
&=E_0\{\bar g_a(y,X)\}\\
&\quad+E_1\left[\frac{\bar\omega(X)\ind(A=a)}{\bar e_a(X)}\{\bar m_a(y,X,S)-\bar g_a(y,X)\}\right]\\
&\quad+E_1\left[\frac{\bar\omega(X)\ind(A=a)M}{\bar e_a(X)\bar\rho_a(X,S)}\{Z_y-\bar m_a(y,X,S)\}\right].
\end{align*}
Let
\[
\bar w_a(x)=\bar\omega(x)\frac{e_a(x)}{\bar e_a(x)},\qquad
\bar r_a(x,s)=\frac{\rho_a(x,s)}{\bar\rho_a(x,s)},
\]
and let $H_a h(x)=E\{h(X,S)\mid R=1,A=a,X=x\}$. The symbol $\bar w_a$ is a ratio-level weight and is deliberately not denoted by $q_a$ to avoid conflict with quantiles.
\end{definition}

Definition \ref{def:signal} records the population version of the one-step signal evaluated at an arbitrary candidate nuisance collection.  The quantities $\bar w_a$ and $\bar r_a$ collect the ratio-level errors in the source-to-target and validation weights.  Writing the drift in terms of these ratios keeps the robustness statement tied to the actual weighted signal used by the estimator.

\begin{proposition}[Exact drift identity and ratio-level robustness]\label{prop:drift}
Under Assumptions \ref{ass:causal}--\ref{ass:positivity}, for any candidate $\bar\eta$ in Definition \ref{def:signal},
\begin{align}
\Psi_a^y(\bar\eta)-\psi_a(y)
&=E_1\left[\{\omega(X)-\bar w_a(X)\}\{\bar g_a(y,X)-g_a(y,X)\}\right] \notag\\
&\quad+E_1\left[\bar w_a(X)H_a\left\{[1-\bar r_a(X,S)]
[\bar m_a(y,X,S)-m_a(y,X,S)]\right\}(X)\right].
\label{eq:drift}
\end{align}
Consequently, $\Psi_a^y(\bar\eta)=\psi_a(y)$ if either $(\bar m_a,\bar g_a)=(m_a,g_a)$ or $(\bar w_a,\bar r_a)=(\omega,1)$ almost surely. If $\bar w_a$ is uniformly bounded, then for a constant $C$ depending only on positivity bounds,
\begin{align*}
\abs{\Psi_a^y(\bar\eta)-\psi_a(y)}
&\le C\norm{\bar w_a-\omega}_{P_{X,1},2}\norm{\bar g_a(y,\cdot)-g_a(y,\cdot)}_{P_{X,1},2}\\
&\quad+C\norm{\bar r_a-1}_{1a,2}\norm{\bar m_a(y,\cdot,\cdot)-m_a(y,\cdot,\cdot)}_{1a,2}.
\end{align*}
\end{proposition}

Proposition \ref{prop:drift} is an exact identity rather than an asymptotic expansion.  It shows that first-order bias disappears when the outcome regressions are correct or when the ratio-level weights are correct, and otherwise is controlled by products of nuisance errors.  This identity is the source of the product-rate conditions imposed below and is the transported-CDF analogue of orthogonal estimating equations used in augmented inverse-probability, doubly robust, and debiased estimation \citep{RobinsRotnitzkyZhao1994,BangRobins2005,ChernozhukovDML2018}.

\subsection{Closed-Form One-Step Estimation}\label{sec:estimation}

The estimator is the empirical analogue of the identified functional plus the two residual corrections appearing in the efficient influence function.  Cross-fitting is used only to separate nuisance training from evaluation of the estimating signal; the final estimator is still a closed-form sample average once the nuisance functions have been fitted.

\begin{definition}[Cross-fitted one-step CDF and quantile estimators]\label{def:estimator}
Let $\mathcal I_1,\ldots,\mathcal I_K$ be a random partition of $\{1,\ldots,n\}$, with fixed $K$. For each fold $k$, estimate $\eta$ using observations outside $\mathcal I_k$; denote the estimate by $\widehat\eta^{(-k)}$. For each observation $i\in\mathcal I_k$, write the fold-specific nuisance evaluation as $\widehat\eta_i=\widehat\eta^{(-k)}$. Define the cross-fitted CDF estimator
\begin{align}
\widehat\psi_a(y)
&=\Pnr{0}\widehat g_{a,i}(y,X_i) \notag\\
&\quad+\Pnr{1}\left[\frac{\widehat\omega_i(X_i)\ind(A_i=a)}{\widehat e_{a,i}(X_i)}
\{\widehat m_{a,i}(y,X_i,S_i)-\widehat g_{a,i}(y,X_i)\}\right] \notag\\
&\quad+\Pnr{1}\left[\frac{\widehat\omega_i(X_i)\ind(A_i=a)M_i}{\widehat e_{a,i}(X_i)\widehat\rho_{a,i}(X_i,S_i)}
\{\ind(Y_i\le y)-\widehat m_{a,i}(y,X_i,S_i)\}\right].
\label{eq:estimator}
\end{align}
For theory, let $\widehat F_a$ be a nondecreasing version of $\widehat\psi_a$ on a compact outcome interval $\cY$ satisfying the first-order conditions stated below; the raw estimator may be used if it is already nondecreasing. Define
\[
\widehat q_a(\tau)=\inf\{y\in\cY:\widehat F_a(y)\ge\tau\},\qquad
\widehat\Delta(\tau)=\widehat q_1(\tau)-\widehat q_0(\tau).
\]
For computation, $\widehat F_a$ may be obtained from the growing-grid and isotonic procedures in Definitions \ref{def:convex} and \ref{def:growing_grid}. A fixed grid is only a numerical summary unless its discretization error is negligible for the inferential scale under consideration.
\end{definition}

The three lines in \eqref{eq:estimator} correspond respectively to target covariate averaging, the source surrogate residual correction, and the validation-outcome residual correction.  The monotone version $\widehat F_a$ is introduced because quantile inversion requires a distribution function, while the first-order theory only requires that the monotonicity correction be asymptotically negligible at the relevant scale.

\begin{definition}[Convex implementation]\label{def:convex}
A fully analyzable implementation uses the following nuisance learners.
\begin{enumerate}[label=(\roman*)]
\item $m_a(y,x,s)$ is estimated among $R=1,A=a,M=1$ units by ridge logistic regression for $\ind(Y\le y)$, or by finite-rank kernel ridge regression (KRR) followed by clipping to $[0,1]$.
\item $g_a(y,x)$ is estimated among $R=1,A=a$ units by regressing a cross-fitted version of $\widehat m_a(y,X,S)$ on $X$ using ridge-linear regression or finite-rank KRR.
\item $e_a(x)$ is known in a randomized trial; otherwise it is estimated among $R=1$ units by ridge logistic regression.
\item $\rho_a(x,s)$ is estimated among $R=1,A=a$ units by ridge logistic regression for $M$.
\item $\omega(x)=dP_{X,0}/dP_{X,1}$ is estimated either by a calibrated ridge logistic classifier or by entropy balancing. If a classifier estimates $p_0(x)=P(R=0\mid X=x)$ on the training sample for fold $k$, the fold-specific density-ratio estimate is
\[
\widehat\omega^{(-k)}(x)=\frac{\widehat p_0^{(-k)}(x)}{1-\widehat p_0^{(-k)}(x)}\frac{\widehat\pi_1^{(-k)}}{\widehat\pi_0^{(-k)}},
\]
truncated to the positivity interval and normalized on the training source sample so that $\mathbb P_{n,1}^{(-k)}\widehat\omega^{(-k)}=1$. This fold-specific normalization preserves cross-fitting independence. Entropy balancing estimates source weights by
\[
\min_{w_i\ge0}\sum_{i:R_i=1}w_i\log w_i
\quad\text{subject to}\quad
\sum_{i:R_i=1}w_i=1,\quad
\sum_{i:R_i=1}w_i b_X(X_i)=\frac1{n_0}\sum_{i:R_i=0}b_X(X_i),
\]
and sets $\widehat\omega^{(-k)}(X_i)=n_1^{(-k)} w_i$ on the training source sample, with an exponential-tilt interpolation for new $x$ when needed; all balancing constraints are imposed using only the training observations for the relevant fold.
\item On a growing grid $\cY_n$, monotonicity is enforced by isotonic projection
\[
(\widehat F_a(y_{1,n}),\ldots,\widehat F_a(y_{J_n,n}))
=\argmin_{0\le f_1\le\cdots\le f_{J_n}\le1}
\sum_{j=1}^{J_n}\{\widehat\psi_a(y_{j,n})-f_j\}^2.
\]
\end{enumerate}
\end{definition}

Definition \ref{def:convex} gives one deliberately restricted implementation path.  The high-level results below do not require these particular learners, but this list makes clear which nuisance-fitting procedures are later verified by low-level rate calculations.  The calibrated density-ratio and entropy-balancing option is the transport-weight analogue of survey calibration, empirical likelihood, and covariate-balancing propensity-score ideas \citep{DevilleSarndal1992,QinLawless1994,ImaiRatkovic2014}.

\begin{proposition}[Closed-form and convex structure]\label{prop:closed_convex}
Given nuisance estimates, \eqref{eq:estimator} is a closed-form sample average. Under Definition \ref{def:convex}, every nuisance subproblem is either a ridge/KRR normal-equation problem or a convex optimization problem; the CDF monotonicity correction is a convex projection computable by the pool-adjacent-violators algorithm; and quantile inversion over a grid is a deterministic search. No nonconvex optimization is required.
\end{proposition}

Proposition \ref{prop:closed_convex} separates the computational claim from the statistical claim.  Statistical validity still depends on whether the fitted nuisance functions make the drift and empirical-process perturbation small enough.  The remaining results in this subsection state those requirements first at a fixed threshold and then uniformly over $y$.

\begin{definition}[Uniform signal classes]\label{def:uniform_signal}
For a nuisance collection $\eta$ and fixed $a$, define the uncentered one-step signal indexed by $y\in\cY$ as
\begin{align*}
\Gamma_a^y(O;\eta)
&=\frac{\ind(R=0)}{\pi_0}g_a(y,X)\\
&\quad+\frac{\ind(R=1)\omega(X)\ind(A=a)}{\pi_1 e_a(X)}\{m_a(y,X,S)-g_a(y,X)\}\\
&\quad+\frac{\ind(R=1)\omega(X)\ind(A=a)M}{\pi_1 e_a(X)\rho_a(X,S)}\{\ind(Y\le y)-m_a(y,X,S)\}.
\end{align*}
Then $E\{\Gamma_a^y(O;\eta)\}=\psi_a(y)$ and $\phi_a^y(O)=\Gamma_a^y(O;\eta)-\psi_a(y)$.  Let $\widehat\Gamma_{a,k}^y$ denote the same signal with fold-specific nuisance estimates trained outside fold $k$ and with $\pi_r$ replaced by the population value; random-denominator corrections are handled separately by Lemma \ref{lem:ratio}. For later bootstrap notation, define the cross-fitted plug-in centered CDF influence function by
\[
\widehat\phi_a^y(O_i)=\widehat\Gamma_{a,k(i)}^y(O_i)-\widehat\psi_a(y),
\]
where $k(i)$ is the fold containing observation $i$.
\end{definition}

The notation $\Gamma_a^y$ packages the uncentered one-step signal so that $\phi_a^y$ is simply its centered oracle version.  This is convenient for asymptotic analysis because the estimator can be decomposed into an oracle empirical process, a population drift, and a nuisance-estimation empirical-process remainder.

\begin{assumption}[High-level nuisance rates for fixed $y$]\label{ass:rates_fixed}
For fixed $(a,y)$ and each fold, the nuisance estimates satisfy the positivity bounds in Assumption \ref{ass:positivity} with probability tending to one and
\begin{align}
&\norm{\widehat\omega\,e_a/\widehat e_a-\omega}_{P_{X,1},2}
\norm{\widehat g_a(y,\cdot)-g_a(y,\cdot)}_{P_{X,1},2}\notag\\
&\quad+\norm{\rho_a/\widehat\rho_a-1}_{1a,2}
\norm{\widehat m_a(y,\cdot,\cdot)-m_a(y,\cdot,\cdot)}_{1a,2}=o_p(n^{-1/2}),\label{eq:productrate}
\end{align}
and the fold-specific one-step signal satisfies
\[
\norm{\widehat\Gamma_a^y-\Gamma_a^y}_{P,2}=o_p(1).
\]
\end{assumption}

Assumption \ref{ass:rates_fixed} is the fixed-threshold version of the requirements suggested by Proposition \ref{prop:drift}.  The product condition controls the population drift, while $L_2(P)$ convergence of the fitted signal controls the cross-fitted empirical-process perturbation.

\begin{theorem}[Asymptotic linearity and efficiency for fixed $y$]\label{thm:fixed_al}
Under Assumptions \ref{ass:causal}--\ref{ass:positivity} and \ref{ass:rates_fixed}, for fixed $(a,y)$,
\[
\sqrt n\{\widehat\psi_a(y)-\psi_a(y)\}
=\frac1{\sqrt n}\sum_{i=1}^n\phi_a^y(O_i)+o_p(1).
\]
Consequently,
\[
\sqrt n\{\widehat\psi_a(y)-\psi_a(y)\}\dto N\{0,V_a(y)\},
\]
and $\widehat\psi_a(y)$ attains the semiparametric efficiency bound in Theorem \ref{thm:eif}.
\end{theorem}

Theorem \ref{thm:fixed_al} gives the fixed-$y$ CDF expansion with the efficient influence function from Theorem \ref{thm:eif}.  To obtain quantile inference, the expansion must hold locally around the true quantile and the inverse-CDF operation must be stable, as in classical and semiparametric QTE arguments \citep{KoenkerBassett1978,Firpo2007}.

\begin{assumption}[Local inverse-map conditions for pointwise quantiles]\label{ass:local_quantile}
Fix $\tau\in\cT$ and let $q_a=q_a(\tau)$. For each $a\in\{0,1\}$, $\widehat F_a$ is nondecreasing on a neighborhood $\mathcal N_a$ of $q_a$, $\widehat q_a(\tau)\in\mathcal N_a$ with probability tending to one, and
\[
\sup_{y\in\mathcal N_a}|\widehat F_a(y)-\psi_a(y)|=o_p(1).
\]
Moreover,
\[
\sup_{y\in\mathcal N_a}\left|\widehat F_a(y)-\psi_a(y)-\Pn\phi_a^y\right|=o_p(n^{-1/2}).
\]
The function $\psi_a$ is differentiable at $q_a$ with derivative $f_a(q_a)>0$. In addition, the local empirical process is asymptotically equicontinuous at the true quantile: for every deterministic sequence $\delta_n\downarrow0$,
\[
\sup_{|y-q_a|\le \delta_n,\,y\in\mathcal N_a}
\left|\Pn(\phi_a^y-\phi_a^{q_a})\right|=o_p(n^{-1/2}).
\]
A sufficient primitive condition for the last display is that $\{\phi_a^y:y\in\mathcal N_a\}$ is locally $P$-Donsker and $\|\phi_a^y-\phi_a^{q_a}\|_{P,2}\to0$ as $y\to q_a$.
\end{assumption}

Assumption \ref{ass:local_quantile} states exactly what is needed for pointwise inversion.  It combines local monotonicity of the estimated CDF, a local linear expansion of $\widehat F_a$, a positive derivative of $\psi_a$ at the target quantile, and a local empirical-process continuity condition for $\phi_a^y$.

\begin{corollary}[Pointwise Bahadur representation and QTE inference]\label{cor:pointwise}
Under Assumptions \ref{ass:causal}--\ref{ass:positivity} and \ref{ass:local_quantile}, for fixed $\tau\in\cT$,
\[
\sqrt n\{\widehat q_a(\tau)-q_a(\tau)\}
=-\frac1{\sqrt n}\sum_{i=1}^n\frac{\phi_a^{q_a(\tau)}(O_i)}{f_a(q_a(\tau))}+o_p(1).
\]
If Assumption \ref{ass:local_quantile} holds for $a=0,1$, then
\[
\sqrt n\{\widehat\Delta(\tau)-\Delta(\tau)\}
=\frac1{\sqrt n}\sum_{i=1}^n\phi_{\Delta,\tau}(O_i)+o_p(1),
\]
and hence
\[
\sqrt n\{\widehat\Delta(\tau)-\Delta(\tau)\}\dto N\{0,\sigma_\Delta^2(\tau)\},
\qquad \sigma_\Delta^2(\tau)=E\{\phi_{\Delta,\tau}^2(O)\}.
\]
\end{corollary}

Corollary \ref{cor:pointwise} is a pointwise Bahadur representation.  It is enough for a fixed quantile level, but not for simultaneous CDF or QTE bands.  The next assumptions strengthen the fixed-$y$ expansion to a uniform expansion on $\cY$, in line with distributional-effect inference based on full counterfactual processes \citep{ChernozhukovFernandezValMelly2013}.

\begin{assumption}[Uniform oracle process regularity]\label{ass:uniform_oracle}
Let $\cY\subset\Rbb$ be compact. For each $a\in\{0,1\}$:
\begin{enumerate}[label=(\roman*)]
\item The functions $y\mapsto m_a(y,x,s)$ and $y\mapsto g_a(y,x)$ are distribution functions for all $(x,s)$ and $x$, respectively, and are uniformly continuous in $y$ in $L_2$:
\[
\sup_{|u-v|\le h}\|m_a(u)-m_a(v)\|_{1a,2}+\sup_{|u-v|\le h}\|g_a(u)-g_a(v)\|_{P_{X,1},2}\to0.
\]
\item The density ratio, treatment propensity, and labeling propensity are bounded as in Assumption \ref{ass:positivity}.
\item The oracle influence-function class $\Phi_a=\{\phi_a^y:y\in\cY\}$ is pointwise measurable and $P$-Donsker with a bounded envelope. Its covariance semimetric is continuous in $y$.
\end{enumerate}
A sufficient primitive condition for (iii) is that the nuisance paths in (i) are contained in bounded one-dimensional VC-subgraph or bounded-variation subgraph classes after multiplication by the bounded weights in \eqref{eq:eif}.
\end{assumption}

Assumption \ref{ass:uniform_oracle} concerns the oracle influence-function class, before nuisance estimation.  It is a regularity condition on the target process itself: the conditional CDF paths must vary continuously in $y$, and the resulting influence-function class must admit a functional central limit theorem on $\cY$.

\begin{lemma}[Oracle CDF empirical process]\label{lem:oracle_donsker}
Under Assumption \ref{ass:uniform_oracle},
\[
\sqrt n(\Pn-P)\phi_a^{\cdot}\dto \mathbb G_a\quad\text{in }\ell^\infty(\cY),
\]
where $\mathbb G_a$ is a tight mean-zero Gaussian process with covariance
\[
\Cov\{\mathbb G_a(y),\mathbb G_a(y')\}=E\{\phi_a^y(O)\phi_a^{y'}(O)\}.
\]
\end{lemma}

Lemma \ref{lem:oracle_donsker} supplies the limiting Gaussian process that would be obtained if the nuisance functions were known.  The role of the next assumption is to ensure that replacing the oracle signal by the cross-fitted estimated signal does not change this first-order process.

\begin{assumption}[Uniform nuisance rates and realized-class entropy]\label{ass:uniform_nuisance}
For each $a$ and each fold, with probability tending to one, the estimated nuisance functions satisfy the same positivity and boundedness bounds as the true nuisance functions. Define
\begin{align*}
r_{w,a,n}&=\norm{\widehat\omega\,e_a/\widehat e_a-\omega}_{P_{X,1},2},\\
r_{g,a,n}&=\sup_{y\in\cY}\norm{\widehat g_a(y,\cdot)-g_a(y,\cdot)}_{P_{X,1},2},\\
r_{\rho\mathrm{rat},a,n}&=\norm{\rho_a/\widehat\rho_a-1}_{1a,2},\\
r_{m,a,n}&=\sup_{y\in\cY}\norm{\widehat m_a(y,\cdot,\cdot)-m_a(y,\cdot,\cdot)}_{1a,2}.
\end{align*}
The rates obey
\[
r_{w,a,n}r_{g,a,n}+r_{\rho\mathrm{rat},a,n}r_{m,a,n}=o_p(n^{-1/2}).
\]
Moreover,
\[
\sup_{y\in\cY}\norm{\widehat\Gamma_{a,k}^y-\Gamma_a^y}_{P,2}=o_p(1),
\]
uniformly over folds $k$. Let
\[
\widehat{\mathcal G}_{a,k}=\{\widehat\Gamma_{a,k}^y-\Gamma_a^y:y\in\cY\}.
\]
Conditional on the training folds, with probability tending to one, $\widehat{\mathcal G}_{a,k}$ is pointwise measurable and has a measurable envelope $\widehat G_{a,k}$ satisfying $\|\widehat G_{a,k}\|_{\infty}=O_p(1)$. There exist deterministic sequences $\delta_{a,n}\downarrow0$, $V_{a,n}\ge1$, and $A_{a,n}\ge e$ such that, with probability tending to one,
\[
\sup_{f\in\widehat{\mathcal G}_{a,k}}\|f\|_{P,2}\le\delta_{a,n},
\]
\[
\sup_Q \log N\{\epsilon\|\widehat G_{a,k}\|_{Q,2},\widehat{\mathcal G}_{a,k},L_2(Q)\}
\le V_{a,n}\log(A_{a,n}/\epsilon),\qquad 0<\epsilon\le1,
\]
and
\[
\delta_{a,n}\sqrt{V_{a,n}\log(A_{a,n}/\delta_{a,n})}
+\frac{V_{a,n}\log(A_{a,n}/\delta_{a,n})}{\sqrt n}=o(1).
\]

Finally, the fold-specific random-denominator classes also satisfy the entropy conditions of Lemma \ref{lem:uniform_ratio}. Specifically, with
\[
\widehat G_{0,a,k}^y(O)=\widehat g_{a,k}(y,X),
\]
and
\begin{align*}
\widehat G_{1,a,k}^y(O)
&=\frac{\widehat\omega_k(X)\ind(A=a)}{\widehat e_{a,k}(X)}\{\widehat m_{a,k}(y,X,S)-\widehat g_{a,k}(y,X)\}\\
&\quad+\frac{\widehat\omega_k(X)\ind(A=a)M}{\widehat e_{a,k}(X)\widehat\rho_{a,k}(X,S)}\{\ind(Y\le y)-\widehat m_{a,k}(y,X,S)\},
\end{align*}
the classes $\{\widehat G_{0,a,k}^y:y\in\cY\}$ and $\{\widehat G_{1,a,k}^y:y\in\cY\}$ obey Lemma \ref{lem:uniform_ratio} for $r=0$ and $r=1$, respectively, with uniformly bounded envelopes. The same is required for the corresponding oracle denominator classes. This additional condition is only used to replace conditional empirical means $\Pnr{r}$ by unconditional empirical means with known $\pi_r$ uniformly in $y$.
This condition contains the fixed-polynomial-entropy case as the special case $V_{a,n}=O(1)$ and $A_{a,n}=O(1)$, but it also allows finite-dimensional sieve classes whose entropy dimension grows slowly with $n$.
\end{assumption}

Assumption \ref{ass:uniform_nuisance} has three parts.  The first is the uniform product-rate drift condition.  The second requires the estimated one-step signal to approach the oracle signal uniformly in $L_2(P)$.  The third controls the entropy of the realized perturbation classes so that cross-fitting can turn $L_2(P)$ convergence into a uniform empirical-process remainder of order $o_p(n^{-1/2})$.

\begin{lemma}[Cross-fitted small empirical process]\label{lem:small_ep}
Let $\widehat{\mathcal G}_{k}$ be a fold-specific random function class independent of the evaluation fold conditional on the training folds. Suppose that, with probability tending to one, it is pointwise measurable, has envelope $\widehat G_k$ with $\|\widehat G_k\|_{\infty}=O_p(1)$, satisfies
\[
\sup_{f\in\widehat{\mathcal G}_{k}}\|f\|_{P,2}\le\delta_n,
\]
and
\[
\sup_Q \log N\{\epsilon\|\widehat G_k\|_{Q,2},\widehat{\mathcal G}_{k},L_2(Q)\}
\le V_n\log(A_n/\epsilon),\qquad 0<\epsilon\le1,
\]
for deterministic sequences $\delta_n\downarrow0$, $V_n\ge1$, $A_n\ge e$ satisfying
\[
\delta_n\sqrt{V_n\log(A_n/\delta_n)}+
\frac{V_n\log(A_n/\delta_n)}{\sqrt n}=o(1).
\]
Then
\[
\sup_{f\in\widehat{\mathcal G}_{k}}\left|(\mathbb P_{n,k}-P)f\right|=o_p(n^{-1/2})
\]
for each fixed number of folds, where $\mathbb P_{n,k}=|\mathcal I_k|^{-1}\sum_{i\in\mathcal I_k}\delta_{O_i}$ denotes the empirical measure over the evaluation fold. The same conclusion holds for $\Pn$ after averaging over a fixed number of folds.
\end{lemma}

Lemma \ref{lem:small_ep} is the empirical-process device that makes the preceding assumption operational.  Conditional on the training folds, the nuisance estimates are fixed relative to the evaluation fold, so the size and entropy of the perturbation class determine whether the fitted-signal empirical process is negligible uniformly in $y$.

\begin{proposition}[Uniform asymptotic linearity of the CDF estimator]\label{prop:uniform_linearization}
Under Assumptions \ref{ass:causal}--\ref{ass:positivity}, \ref{ass:uniform_oracle}, and \ref{ass:uniform_nuisance},
\[
\sup_{y\in\cY}\left|\widehat\psi_a(y)-\psi_a(y)-\Pn\phi_a^y\right|=o_p(n^{-1/2}).
\]
\end{proposition}

Proposition \ref{prop:uniform_linearization} is the uniform analogue of Theorem \ref{thm:fixed_al}.  It states that, over the whole compact outcome interval $\cY$, the estimated CDF behaves like the empirical mean of the oracle efficient influence function.

\begin{theorem}[Uniform CDF inference]\label{thm:uniform}
If the conditions of Proposition \ref{prop:uniform_linearization} hold for a fixed arm $a$, then
\[
\sqrt n\{\widehat\psi_a-\psi_a\}\dto \mathbb G_a
\quad\text{in }\ell^\infty(\cY),
\]
where $\mathbb G_a$ is the Gaussian process in Lemma \ref{lem:oracle_donsker}. If the same conditions hold jointly for $a=0,1$, then, for the distributional treatment effect,
\[
\sqrt n\{(\widehat\psi_1-\widehat\psi_0)-(\psi_1-\psi_0)\}\dto \mathbb G_1-\mathbb G_0
\quad\text{in }\ell^\infty(\cY),
\]
with the covariance induced by the joint influence functions $(\phi_1^y,\phi_0^{y'})$.
\end{theorem}

Theorem \ref{thm:uniform} converts the uniform linearization into weak convergence of the CDF process and the distributional treatment-effect process.  Once this process convergence is available, quantile-process convergence follows by the same inverse-CDF argument used pointwise, but now uniformly over $\cT$.

\begin{corollary}[Uniform quantile and QTE inference]\label{cor:uniform_quantile}
Suppose the conditions of Theorem \ref{thm:uniform} hold for $a=0,1$, Assumption \ref{ass:quantile} holds, and $\{q_a(\tau):\tau\in\cT\}\subset\cY$. Suppose $\widehat F_a$ is a nondecreasing version of $\widehat\psi_a$ satisfying
\[
\sup_{y\in\cY}|\widehat F_a(y)-\widehat\psi_a(y)|=o_p(n^{-1/2}).
\]
Then
\[
\sqrt n\{\widehat q_a-q_a\}\dto
-\frac{\mathbb G_a(q_a(\cdot))}{f_a(q_a(\cdot))}
\quad\text{in }\ell^\infty(\cT),
\]
and
\[
\sqrt n\{\widehat\Delta-\Delta\}\dto
-\frac{\mathbb G_1(q_1(\cdot))}{f_1(q_1(\cdot))}
+\frac{\mathbb G_0(q_0(\cdot))}{f_0(q_0(\cdot))}
\quad\text{in }\ell^\infty(\cT).
\]
For a growing-grid implementation, the same conclusion holds if the grid and projection conditions of Proposition \ref{prop:grid_projection} hold.
\end{corollary}

Corollary \ref{cor:uniform_quantile} is the main simultaneous-inference result for the QTE curve.  Wald-type simultaneous bands require an estimated version of the QTE influence function and, therefore, estimates of the local density factors appearing in Corollary \ref{cor:eif_quantile}.  The next result states these requirements at a high level.

\begin{corollary}[High-level multiplier bootstrap simultaneous bands]\label{cor:multiplier}
Under the conditions of Corollary \ref{cor:uniform_quantile}, suppose additionally that the following high-level estimated-process conditions hold. Define
\[
\widehat\phi_{\Delta,\tau}(O_i)
=-\frac{\widehat\phi_1^{\widehat q_1(\tau)}(O_i)}{\widehat f_1\{\widehat q_1(\tau)\}}
+\frac{\widehat\phi_0^{\widehat q_0(\tau)}(O_i)}{\widehat f_0\{\widehat q_0(\tau)\}},
\]
where $\widehat\phi_a^y$ is the plug-in centered CDF influence function in Definition \ref{def:uniform_signal}. The estimated QTE influence-function class
\[
\widehat\Phi_\Delta=\{\widehat\phi_{\Delta,\tau}:\tau\in\cT\}
\]
is pointwise measurable with probability tending to one and has a uniformly bounded envelope. Conditional on the data, $\widehat\Phi_\Delta$ satisfies a VC-type entropy bound of the form
\[
\sup_Q\log N\{\epsilon\|\widehat G_\Delta\|_{Q,2},\widehat\Phi_\Delta,L_2(Q)\}
\le V_{\Delta,n}\log(A_{\Delta,n}/\epsilon),\qquad 0<\epsilon\le1,
\]
for a measurable envelope $\widehat G_\Delta$ with $\|\widehat G_\Delta\|_\infty=O_p(1)$ and deterministic sequences obeying the corresponding multiplier maximal-inequality requirement
\[
\sup_{\tau\in\cT}\|\widehat\phi_{\Delta,\tau}-\phi_{\Delta,\tau}\|_{P,2}=o_p(1).
\]
Assume the estimated density factors satisfy $\inf_{a,\tau}\widehat f_a\{\widehat q_a(\tau)\}\ge c_f/2$ with probability tending to one and
\[
\sup_{a,\tau\in\cT}\left|\widehat f_a\{\widehat q_a(\tau)\}-f_a\{q_a(\tau)\}\right|=o_p(1).
\]
Let $\xi_1,\ldots,\xi_n$ be i.i.d. multipliers with mean zero, variance one, and finite exponential moment, independent of the data. Then, conditionally on the data,
\[
\mathbb Z_n^*(\tau)=\frac1{\sqrt n}\sum_{i=1}^n\xi_i\{\widehat\phi_{\Delta,\tau}(O_i)-\Pn\widehat\phi_{\Delta,\tau}\}
\]
converges weakly in probability in $\ell^\infty(\cT)$ to the limiting QTE Gaussian process in Corollary \ref{cor:uniform_quantile}. If the distribution of the supremum norm of the limiting process is continuous at its $(1-\alpha)$ quantile, bootstrap critical values give asymptotically valid simultaneous confidence bands for $\Delta(\tau)$ on $\cT$. Alternatively, one may invert valid multiplier CDF bands; that route avoids explicit density estimation but targets bands constructed through the inverse-map operation rather than Wald-type QTE bands.
\end{corollary}

The multiplier-bootstrap statement deliberately separates process approximation from density estimation.  It gives valid simultaneous Wald-type QTE bands when the estimated influence-function class is close to the oracle class, using the same type of Gaussian and multiplier approximation logic used for suprema of empirical processes \citep{ChernozhukovChetverikovKato2014}.  The alternative inversion of CDF bands avoids explicit density estimation but produces bands through the inverse-map operation rather than through a direct Wald expansion for $\Delta(\tau)$.

\begin{definition}[Growing-grid implementation]\label{def:growing_grid}
Let $\cY_n=\{y_{1,n}<\cdots<y_{J_n,n}\}\subset\cY$ be a deterministic grid with mesh
\[
h_n=\sup_{y\in\cY}\min_{1\le j\le J_n}|y-y_{j,n}|.
\]
The raw grid estimator linearly interpolates $\widehat\psi_a(y_{j,n})$ between adjacent grid points. Quantile inversion is applied to a nondecreasing version $\widehat F_{a,n}$, obtained either because the raw interpolant is already nondecreasing or by isotonic projection on the grid followed by interpolation. For root-$n$ quantile inference we require $h_n=o(n^{-1/2})$ or another condition implying that interpolation and discretization errors are $o(n^{-1/2})$ on the relevant quantile region, and we require the monotonicity correction to be first-order negligible as stated in Proposition \ref{prop:grid_projection}.
\end{definition}

The preceding process statements are written for a continuum of outcome thresholds.  In computation, the CDF is evaluated on a grid and then interpolated or projected.  Definition \ref{def:growing_grid} records the condition under which this numerical representation is fine enough for root-$n$ quantile inference.

\begin{proposition}[Grid approximation and isotonic projection]\label{prop:grid_projection}
Assume $\psi_a$ is continuously differentiable on $\cY$ with derivative bounded above by $C_f$ and bounded below by $c_f>0$ on the quantile region. Let $\widetilde\psi_{a,n}$ be the linear interpolant of the unprojected grid values $\widehat\psi_a(y_{j,n})$. Suppose
\[
\max_{1\le j\le J_n}\left|\widehat\psi_a(y_{j,n})-\psi_a(y_{j,n})-\Pn\phi_a^{y_{j,n}}\right|=o_p(n^{-1/2}),
\]
\[
\sup_{|u-v|\le h_n}\left|\Pn(\phi_a^u-\phi_a^v)\right|=o_p(n^{-1/2}),
\qquad h_n=o(n^{-1/2}).
\]
Then
\[
\sup_{y\in\cY}\left|\widetilde\psi_{a,n}(y)-\psi_a(y)-\Pn\phi_a^y\right|=o_p(n^{-1/2}).
\]
Let $\widehat F_{a,n}$ be a nondecreasing version of $\widetilde\psi_{a,n}$ on the quantile region, for example the isotonic projection on the grid followed by interpolation, and suppose
\[
\sup_{y\in\cY_q}|\widehat F_{a,n}(y)-\widetilde\psi_{a,n}(y)|=o_p(n^{-1/2})
\]
for a compact interval $\cY_q\subset\cY$ containing $\{q_a(\tau):\tau\in\cT\}$ in its interior. Then the generalized inverse of $\widehat F_{a,n}$ satisfies the pointwise and uniform inverse-map limits in Corollary \ref{cor:pointwise} and Corollary \ref{cor:uniform_quantile}. On a fixed grid with strictly increasing true grid values, isotonic projection is asymptotically inactive. On a growing grid, the displayed projection-distance condition is required; the expansion of the raw interpolant alone does not imply monotonicity. In all cases, isotonic projection is a convex nonexpansive Euclidean projection and is a finite-sample stabilizer.
\end{proposition}

Proposition \ref{prop:grid_projection} explains why the grid and isotonic projection steps do not alter the leading asymptotic distribution when their approximation errors are below the root-$n$ scale.  This completes the high-level inference theory.  The next subsection gives sufficient lower-level conditions for the nuisance-rate and entropy assumptions used above.

\subsection{Lower-Level Nuisance-Rate Verification}\label{sec:nuisance_rates}

The preceding results are stated in terms of high-level product rates and realized-class entropy.  This subsection verifies those conditions for a restricted class of estimators whose behavior can be controlled directly.  The point is not to cover all possible machine-learning algorithms, but to show that the required rates are attainable under transparent approximation, dimension, and regularity conditions, paralleling nonparametric and series-estimation theory as well as high-dimensional orthogonal-estimation analyses \citep{Stone1980,Stone1985,Newey1997,BelloniChernozhukovFernandezValHansen2017,ChernozhukovDML2018}.

\begin{definition}[Verified convex finite-dimensional nuisance learners]\label{def:sieve_learners}
Let $b_{X,n}:\cX\to\Rbb^{d_{X,n}}$ and $b_{U,n}:\cX\times\cS\to\Rbb^{d_{U,n}}$ be bounded feature maps, with $U=(X,S)$. A verified convex finite-dimensional implementation is one of the following specific implementations; the theorem below does not cover arbitrary flexible convex learners:
\begin{enumerate}[label=(\roman*)]
\item $m_a(y,x,s)$ is estimated over $R=1,A=a,M=1$ by ridge logistic regression for $\ind(Y\le y)$ or by ridge least squares/KRR in $b_{U,n}(x,s)$ followed by clipping to $[0,1]$;
\item $g_a(y,x)$ is estimated over $R=1,A=a$ by ridge least squares/KRR in $b_{X,n}(x)$ using cross-fitted $\widehat m_a(y,X,S)$ as the response;
\item $e_a(x)$ is known by design or estimated by ridge logistic regression in $b_{X,n}(x)$;
\item $\rho_a(x,s)$ is estimated by ridge logistic regression in $b_{U,n}(x,s)$;
\item $\omega(x)$ is estimated either by a ridge logistic source-target classifier in $b_{X,n}(x)$ or by entropy balancing over $b_{X,n}(x)$ under the exponential-tilt density-ratio model and regularity conditions stated in Assumption \ref{ass:sieve_regular}.
\end{enumerate}
When $g_a$ is fit with generated responses $\widehat m_a(y,X_i,S_i)$, nested cross-fitting means that the source training sample used for the second-stage $g_a$ regression is split into inner folds and each generated response is obtained from an $m_a$ estimator trained without observation $i$. Thus the second-stage regression noise is conditionally independent of the observation on which the generated response is evaluated. All empirical risks are normalized as sample averages and use ridge penalties $\lambda_n\|\theta\|_2^2/2$.
\end{definition}

Definition \ref{def:sieve_learners} fixes the class of finite-dimensional procedures to which the subsequent rate theorem applies.  The generated-regressor step for $g_a$ is treated explicitly because $g_a$ is learned from estimates of $m_a$ rather than from directly observed responses.

\begin{assumption}[Convex sieve regularity]\label{ass:sieve_regular}
For the learners in Definition \ref{def:sieve_learners}, assume:
\begin{enumerate}[label=(\roman*)]
\item $\|b_{X,n}(X)\|_2\le C\sqrt{d_{X,n}}$ and $\|b_{U,n}(U)\|_2\le C\sqrt{d_{U,n}}$ almost surely, and the corresponding population Gram or logistic Hessian matrices have eigenvalues bounded above and below by positive constants uniformly in $n$ and $y\in\cY$. Ridge penalties are strictly positive when needed for numerical uniqueness; if an empirical convex problem has multiple minimizers, any minimizer with the same fitted values on the relevant design span may be used.
\item The population minimizers of the working risks are unique in the relevant fitted-function metric, and the empirical minimizers lie in the corresponding local neighborhoods with probability tending to one.
\item The best sieve approximations to $m_a$, $g_a$, $e_a$, $\rho_a$, and $\omega$ have $L_2$ errors $a_{m,n}$, $a_{g,n}$, $a_{e,n}$, $a_{\rho,n}$, and $a_{\omega,n}$, respectively, uniformly over $a$ and $y\in\cY$ where applicable.
\item The coefficient vectors of the best approximations have Euclidean norms bounded by $B_n$, and $\lambda_n B_n=o(r_n)$ for the relevant rate $r_n$ below.
\item $d_{X,n}\log n/n\to0$, $d_{U,n}\log n/n\to0$, and the empirical Hessians are uniformly consistent on neighborhoods of the best approximating parameters.
\item For continuum $\cY$, the maps $y\mapsto m_a(y,\cdot)$ and $y\mapsto g_a(y,\cdot)$ are uniformly Lipschitz in $L_2$, and all $y$-indexed regressions are implemented on a grid with mesh $h_n$ and linear interpolation. For a fixed finite grid, set $h_n=0$.
\item If entropy balancing is used for $\omega$, the true or pseudo-true density ratio lies in an exponential-tilt sieve $\omega_{\theta,n}(x)=c(\theta)\exp\{\theta^\top b_{X,n}(x)\}$ up to approximation error $a_{\omega,n}$, the normalizing constant $c(\theta)$ is chosen so that $E_1\omega_{\theta,n}(X)=1$, the coefficient norm is bounded by $B_n$, and the population Jacobian of the calibration moments has eigenvalues bounded above and below by positive constants. The same bounded-feature and entropy-growth conditions used for the other finite-dimensional sieves apply to this class when it is evaluated at new $x$ values.
\end{enumerate}
\end{assumption}

Assumption \ref{ass:sieve_regular} supplies the usual ingredients for finite-dimensional convex M-estimation: bounded features, stable population curvature, controlled approximation error, and manageable growth of feature dimension and threshold grid size.  The entropy-balancing clause gives an analogous local identification condition for estimating $\omega$ through calibration.

\begin{lemma}[Uniform convex M-estimation rate]\label{lem:uniform_mrate}
Let $\widehat\theta_y$ minimize a ridge-penalized convex empirical risk indexed by $y\in\cY_n$, and let $\theta_y^\star$ be the corresponding population minimizer. Under Assumption \ref{ass:sieve_regular},
\[
\sup_{y\in\cY_n}\|\widehat\theta_y-\theta_y^\star\|_2
=O_p\left(\sqrt{\frac{d_n+\log(eJ_n)}{n}}+\lambda_n B_n\right),
\]
where $d_n$ is the relevant feature dimension and $J_n=|\cY_n|$; if $J_n=1$, the term is interpreted as $\sqrt{d_n/n}$. The same rate holds for the induced $L_2$ fitted-function error, plus the sieve approximation error. For continuum $\cY$, an additional interpolation error of order $h_n$ is incurred under the Lipschitz condition.
\end{lemma}

Lemma \ref{lem:uniform_mrate} is the basic rate statement used repeatedly for the nuisance functions.  It gives a uniform coefficient and fitted-function rate over the grid $\cY_n$, with interpolation error added when the target is a continuum of thresholds.

\begin{theorem}[Uniform rates for convex sieve implementation]\label{thm:sieve_rates}
Suppose Assumptions \ref{ass:positivity} and \ref{ass:sieve_regular} hold and the two-stage estimator of $g_a$ uses the nested cross-fitting scheme in Definition \ref{def:sieve_learners}, so each generated response used in the second-stage risk is trained without that observation and without the outer evaluation fold. Define
\begin{align*}
r_{m,n}&=a_{m,n}+\sqrt{\frac{d_{U,n}+\log(eJ_n)}{n}}+h_n,\\
r_{e,n}&=a_{e,n}+\sqrt{\frac{d_{X,n}}{n}},\\
r_{\rho,n}&=a_{\rho,n}+\sqrt{\frac{d_{U,n}}{n}},\\
r_{\omega,n}&=a_{\omega,n}+\sqrt{\frac{d_{X,n}}{n}},\\
r_{g,n}&=a_{g,n}+r_{m,n}+\sqrt{\frac{d_{X,n}+\log(eJ_n)}{n}}+h_n.
\end{align*}
Then the convex sieve learners satisfy, uniformly over $a$ and $y\in\cY$ where applicable,
\begin{align*}
\sup_{y\in\cY}\|\widehat m_a(y)-m_a(y)\|_{1a,2}&=O_p(r_{m,n}),\\
\sup_{y\in\cY}\|\widehat g_a(y)-g_a(y)\|_{P_{X,1},2}&=O_p(r_{g,n}),\\
\|\widehat e_a-e_a\|_{P_{X,1},2}&=O_p(r_{e,n}),\\
\|\widehat\rho_a-\rho_a\|_{1a,2}&=O_p(r_{\rho,n}),\\
\|\widehat\omega-\omega\|_{P_{X,1},2}&=O_p(r_{\omega,n}).
\end{align*}
Moreover, if fitted propensities and density ratios are truncated to the positivity interval and each fold-specific $\widehat\omega^{(-k)}$ is normalized by $\mathbb P_{n,1}^{(-k)}\widehat\omega^{(-k)}=1$, then
\begin{align*}
r_{w,a,n}&=O_p(r_{\omega,n}+r_{e,n}),\\
r_{\rho\mathrm{rat},a,n}&=O_p(r_{\rho,n}).
\end{align*}
\end{theorem}

Theorem \ref{thm:sieve_rates} translates the generic convex M-estimation rate into rates for $m_a$, $g_a$, $e_a$, $\rho_a$, and $\omega$.  The rate for $g_a$ includes the first-stage error from estimating $m_a$, reflecting the nested structure of the surrogate-integrated regression.

\begin{corollary}[Sufficient low-level product-rate and entropy conditions]\label{cor:sieve_products}
Under the conditions of Theorem \ref{thm:sieve_rates}, define
\[
D_n=d_{X,n}+d_{U,n}+\log(eJ_n),\qquad
r_{\Gamma,n}=r_{m,n}+r_{g,n}+r_{e,n}+r_{\rho,n}+r_{\omega,n}.
\]
Assumption \ref{ass:uniform_nuisance} holds if
\[
(r_{\omega,n}+r_{e,n})r_{g,n}+r_{\rho,n}r_{m,n}=o(n^{-1/2}),
\]
\[
r_{\Gamma,n}\sqrt{D_n\log n}+\frac{D_n\log n}{\sqrt n}=o(1),
\]
and the grid/interpolation error is included in $r_{\Gamma,n}$. In particular, if $e_a$, $\rho_a$, and $\omega$ are known by design or estimated at $n^{-1/2}$ rate and $r_{m,n}+r_{g,n}=o(1)$ with $(r_{m,n}+r_{g,n})\sqrt{D_n\log n}+D_n\log n/\sqrt n=o(1)$, then the product-rate and perturbation-process conditions hold. If all nuisance functions have a common rate $r_n$ and $D_n$ is bounded, the familiar sufficient condition is $r_n=o(n^{-1/4})$; if $D_n\to\infty$, the additional displayed growing-entropy condition is required.
\end{corollary}

Corollary \ref{cor:sieve_products} is the bridge back to Assumption \ref{ass:uniform_nuisance}.  It shows that the high-level drift and empirical-process conditions can be checked by combining approximation errors, feature dimensions, and the outcome-grid size.  When the dimension is fixed, the product-rate condition reduces to the familiar fourth-root requirement; when the dimension grows, entropy growth must also be controlled.

\begin{assumption}[Uniform RKHS-KRR regularity]\label{ass:rkhs_uniform}
For each KRR regression nuisance, the kernel is bounded and measurable, the response is bounded, and the target regression functions $f_y$ indexed by $y\in\cY$ belong uniformly to the source class $\mathcal R(T^r)$ of the kernel integral operator $T$ with radius bounded by a constant. Here $T$ is the integral operator under the covariate law relevant to the particular regression problem; separate instances are used for $m_a$ under the law of $(X,S)\mid R=1,A=a,M=1$ and for $g_a$ under the law of $X\mid R=1,A=a$. The effective dimension satisfies
\[
\mathcal N(\lambda)=\operatorname{tr}\{T(T+\lambda I)^{-1}\}\le C\lambda^{-\alpha},\qquad \alpha\in(0,1].
\]
The map $y\mapsto f_y$ is Lipschitz in $L_2(P)$ and in the RKHS interpolation norm used by the source condition. KRR is fit on a grid $\cY_n$ and linearly interpolated, with regularization $\lambda_n\asymp n^{-1/(2r+\alpha)}$.
\end{assumption}

Assumption \ref{ass:rkhs_uniform} gives an alternative route for regression-type nuisance functions when kernel ridge regression is used.  The source condition and effective-dimension bound determine the regression rate, while the Lipschitz condition in $y$ makes the grid-based implementation compatible with the uniform CDF theory.

\begin{theorem}[Uniform RKHS-KRR rate verification for regression nuisances]\label{thm:rkhs_rates}
Under Assumption \ref{ass:rkhs_uniform}, suppose the effective sample size for a KRR regression nuisance is $N_{\mathrm{eff}}$ and $N_{\mathrm{eff}}/n$ is bounded away from zero with probability tending to one. This applies to the regression-type nuisances $m_a$ and $g_a$ under Assumption \ref{ass:positivity}, with nested cross-fitting for the generated response in $g_a$. Then
\[
\sup_{y\in\cY}\|\widehat f_y-f_y\|_{P,2}
=O_p\left(N_{\mathrm{eff}}^{-r/(2r+\alpha)}\sqrt{\log(eJ_n)}+h_n\right).
\]
Consequently, this theorem verifies the regression $L_2$ rates for $m$ and $g$. If $m$ and $g$ are estimated by KRR at rate $r_{K,n}=n^{-r/(2r+\alpha)}\sqrt{\log(eJ_n)}+h_n$ and the nuisance rates for $(e,\rho,\omega)$ are separately verified as $s_{e,n},s_{\rho,n},s_{\omega,n}$ by logistic, calibration, or density-ratio arguments, then the product-rate part of Assumption \ref{ass:uniform_nuisance} holds whenever
\[
(s_{\omega,n}+s_{e,n})r_{K,n}+s_{\rho,n}r_{K,n}=o(n^{-1/2}).
\]
The perturbation-process part of Assumption \ref{ass:uniform_nuisance} additionally requires the entropy-growth condition in Corollary \ref{cor:sieve_products}, with $r_{\Gamma,n}$ formed from $r_{K,n}$ and the separately verified rates. If $(e,\rho,\omega)$ are parametric or known and the KRR classes have bounded entropy dimension, any $r_{K,n}=o(1)$ suffices for the product term. If all nuisances have a common KRR-type rate $r_{K,n}$ and the corresponding density-ratio/classification rates are separately justified, a sufficient smoothness condition for the product term is $r/(2r+\alpha)>1/4$ up to logarithmic and grid-mesh factors.
\end{theorem}

Theorem \ref{thm:rkhs_rates} verifies the regression-rate part of the nuisance conditions for KRR.  It does not by itself verify density-ratio, treatment, or validation propensity estimation; those components must still be justified by the corresponding classifier, calibration, or propensity-rate arguments.

\subsection{Surrogate Efficiency Gains}\label{sec:eff_gain}

The preceding results establish efficiency of the surrogate-assisted estimator in the observed-data model.  This subsection asks a narrower comparison question: when is observing $S$ among units with missing $Y$ more informative than a benchmark that does not observe those surrogates?  To make this comparison meaningful, both procedures must identify the same target functional under a common model.

A comparison with a procedure that does not observe $S$ on units with missing $Y$ requires a common identifying model. Without further restrictions, if $M$ depends on $S$, then dropping $S$ from the missing-outcome units changes the observed-data model and may destroy identification. Therefore the following assumption is used only for the benchmark comparison.

\begin{assumption}[Surrogate-ignorable labeling for the no-surrogate benchmark]\label{ass:rho_x}
For each $a\in\{0,1\}$,
\[
\rho_a(X,S)=\rho_a^0(X)=P(M=1\mid R=1,A=a,X).
\]
\end{assumption}

Assumption \ref{ass:rho_x} is not needed for the main identification or efficiency results.  It is imposed only so that the no-surrogate benchmark remains identifiable after dropping $S$ for units with $M=0$.  When validation depends on $S$, the benchmark and the surrogate-assisted model no longer represent the same observed-data problem.

\begin{definition}[No-surrogate benchmark]\label{def:noS}
Under Assumption \ref{ass:rho_x}, define the benchmark observed-data model in which $S$ is not observed for source units with $M=0$. Let
\[
g_{a}^{\mathrm{nos}}(y,x)=P(Y\le y\mid R=1,A=a,X=x,M=1).
\]
Then $g_{a}^{\mathrm{nos}}(y,x)=g_a(y,x)$ under Assumption \ref{ass:rho_x}. The benchmark efficient influence function is
\[
\phi_{a,0}^y(O)=
\frac{\ind(R=0)}{\pi_0}\{g_a(y,X)-\psi_a(y)\}
+\frac{\ind(R=1)\omega(X)\ind(A=a)M}{\pi_1 e_a(X)\rho_a^0(X)}
\{Z_y-g_a(y,X)\}.
\]
Let $V_{a,0}(y)=E\{[\phi_{a,0}^y(O)]^2\}$.
\end{definition}

Definition \ref{def:noS} gives the influence function for the benchmark model.  Compared with Definition \ref{def:eif}, the benchmark cannot use the residual $m_a-g_a$ among source units with missing $Y$ because $S$ is unavailable there.  The variance comparison therefore isolates the information contributed by observing the surrogate on those otherwise unlabeled source units.

\begin{theorem}[Distributional efficiency gain from observing surrogates]\label{thm:eff_gain_cdf}
Suppose Assumptions \ref{ass:causal}, \ref{ass:mar}, \ref{ass:transport}, \ref{ass:positivity}, and \ref{ass:rho_x} hold. Then
\[
V_{a,0}(y)-V_a(y)
=\frac{1}{\pi_1}E_1\left[
\frac{\omega^2(X)}{e_a(X)}\frac{1-\rho_a^0(X)}{\rho_a^0(X)}
\V\{m_a(y,X,S)\mid R=1,A=a,X\}
\right]\ge0.
\]
The inequality is strict whenever the displayed conditional variance is positive on a set with positive $P_{X,1}$ probability and $\rho_a^0(X)<1$ on that set.
\end{theorem}

Theorem \ref{thm:eff_gain_cdf} shows that the gain is nonnegative and is driven by the conditional variation of $m_a(y,X,S)$ given $(X,A=a,R=1)$.  Thus $S$ improves efficiency at threshold $y$ exactly to the extent that it predicts the threshold outcome after conditioning on $X$, and the gain is larger when validation is scarcer through the factor $(1-\rho_a^0)/\rho_a^0$.

\begin{corollary}[Quantile-specific efficiency gain]\label{cor:eff_gain_quantile}
Under the assumptions of Theorem \ref{thm:eff_gain_cdf} and Assumption \ref{ass:quantile}, the efficiency gain for estimating $q_a(\tau)$ from observing $S$ is
\[
\frac{V_{a,0}(q_a(\tau))-V_a(q_a(\tau))}{\{f_a(q_a(\tau))\}^2}.
\]
The gain for the QTE $\Delta(\tau)$ is
\[
\sum_{a=0}^1
\frac{V_{a,0}(q_a(\tau))-V_a(q_a(\tau))}{\{f_a(q_a(\tau))\}^2}.
\]
Thus a surrogate may substantially improve inference at one quantile and provide little gain at another.
\end{corollary}

Corollary \ref{cor:eff_gain_quantile} transfers the CDF-level variance reduction to quantiles through the inverse-CDF derivative.  The same surrogate can therefore yield different gains at different quantiles, because both the predictive content of $m_a(q_a(\tau),X,S)$ and the density factor $f_a(q_a(\tau))$ may vary with $\tau$.

\section{Simulation Studies}\label{sec:experiments}

\subsection{Design of Simulation Studies}\label{sec:sim_design}

The simulation study evaluates the finite-sample behavior of the transported distributional and quantile treatment-effect procedures developed in Section \ref{sec:theory}.  Throughout, the observed data have the same structure as in the main analysis: target units contribute baseline covariates only, while source units contribute treatment, surrogate, validation status, and the primary outcome when validated.  The target estimand is always the transported QTE
\[
\Delta(\tau)=q_1(\tau)-q_0(\tau),
\qquad \tau\in\{0.25,0.50,0.75\}.
\]

In all experiments, source and target samples are of comparable size.  The source covariates are generated on a bounded support, and the target covariate law is an exponential tilt of the source law, so that the density ratio \(\omega(x)=dP_{X,0}/dP_{X,1}(x)\) is known for oracle calculations but must be estimated by the feasible procedures.  Source treatment is randomized in the experiments reported here.  Potential surrogates and outcomes are generated from treatment-specific models of the form
\[
S^a=h_a(X)+\varepsilon_S,
\qquad
Y^a=\nu_a(X)+\lambda_S\beta_a^\top S^a+\sigma_a(X)\varepsilon_Y,
\]
with Gaussian errors.  The details of \(h_a,\nu_a,\sigma_a\), and \(\beta_a\) differ across experiments to target different theoretical claims; the replication code fixes these functions across Monte Carlo replications.

The true target CDFs, quantiles, and QTEs are computed from a large independent target Monte Carlo sample using the closed-form conditional CDFs implied by the Gaussian model. Feasible estimators use five-fold sample splitting. Nuisance functions are fit with stable ridge and logistic working regressions, and estimated probabilities and density ratios are truncated only to avoid numerical instability. We report Monte Carlo bias, MSE, RMSE, and, for one-step estimators with implemented influence-function standard errors, pointwise coverage and interval length.  Unless otherwise stated, results are based on \(B_{\mathrm{MC}}=500\) replications.

\paragraph{Experiment 1: pointwise finite-sample behavior under transport and surrogate-dependent validation.} The first experiment is a stress test of the full transported missing-primary-outcome problem.  It combines source-target covariate shift, a nonlinear surrogate process, and validation probabilities that depend on the post-treatment surrogate.  The conditional outcome law given \((X,S,A)\) is intentionally learnable, so that \(m_a(y,X,S)\) is not made difficult artificially; however, the induced \(g_a(y,X)=E\{m_a(y,X,S)\mid R=1,A=a,X\}\) is nonlinear.  This makes a low-dimensional transported plug-in estimator vulnerable to persistent approximation bias, while the proposed one-step estimator can use both the source surrogate-process residual and the validated-outcome residual to correct it.

The main competitors are SA, the proposed cross-fitted surrogate-assisted one-step estimator; Oracle, the same estimating equation with true nuisance functions; IPW, a validation-only transported inverse-probability weighted estimator; Plugin, the transported regression plug-in estimator without one-step correction; and Source, a negative control that averages over the source rather than target covariate law.  The main table uses \(n=2000\); a larger replication with \(n=4000\) is reported in Appendix Table \ref{tab:exp1_appendix_n4000}.

\paragraph{Experiment 2: efficiency gains from observing surrogates.} The second experiment isolates the efficiency-gain claim in Theorem \ref{thm:eff_gain_cdf} and Corollary \ref{cor:eff_gain_quantile}.  To make the comparison with the no-surrogate benchmark clean, this experiment uses a correctly specified linear-Gaussian specialization and sets the validation probability to satisfy Assumption \ref{ass:rho_x}; in the reported implementation, \(\rho_a^0(X)\equiv\bar\rho\).  Thus SA and NoS identify the same transported QTE, and their difference reflects the information carried by the observed surrogate among source units whose primary outcome is missing.

We vary the surrogate-primary-outcome association and the primary-outcome validation rate over a small grid,
\[
\lambda_S\in\{0,1,2\},
\qquad
\bar\rho\in\{0.20,0.40,0.70\},
\]
while keeping the transport and randomized treatment mechanisms fixed. 

SA is compared with NoS, an otherwise analogous transported one-step estimator that targets the same transported QTE but does not use \(S\). 
Specifically, SA uses the surrogate-assisted regression \(m_a(y,X,S)\) and the source surrogate-process correction \(m_a(y,X,S)-g_a(y,X)\), whereas NoS works only with the \(X\)-level regression \(g_a(y,X)\) and the validation residual among units with observed primary outcomes. 

We also report IPW, Oracle, and FullOracle. IPW is a validation-only transported weighting estimator. Oracle is an infeasible observed-data benchmark that uses the true nuisance functions under the same missing-primary-outcome structure as SA. FullOracle is an infeasible full-data benchmark that treats all source primary outcomes as observed.

\paragraph{Experiment 3: source-target covariate shift and density-ratio adjustment.}
The third experiment isolates the transport component of the efficient influence function. 
Whereas Experiment 1 evaluates the full transported missing-outcome problem and Experiment 2 isolates the efficiency gain from observing surrogates, this experiment varies only the strength of source-target covariate shift. 
We fix the surrogate-outcome association, validation rate, and validation-surrogate dependence, and vary
\[
c\in\{0,0.4,0.8,1.2\}.
\]
Both randomized and observational source treatment assignment mechanisms are considered.

The proposed estimator with estimated \(\widehat\omega\) is compared with the same estimator using the true density ratio, a source-population negative control, and a validation-only transported IPW estimator. 
The source negative control estimates the source-population QTE rather than the target-population QTE. 
Thus it should be nearly unbiased when \(c=0\), but should become biased as the target and source covariate laws diverge. 
By contrast, the transported SA estimator should remain approximately unbiased when the density-ratio and regression learners are adequate. 
We also report the effective source sample size induced by the fitted density ratio,
\[
\mathrm{ESS}_{\omega}
=
\frac{\left(\sum_{i:R_i=1}\widehat\omega(X_i)\right)^2}
{\sum_{i:R_i=1}\widehat\omega^2(X_i)},
\]
as an overlap diagnostic. 
Comparing estimated-\(\omega\) and oracle-\(\omega\) versions separates intrinsic overlap deterioration from density-ratio estimation error.

\paragraph{Experiment 4: growing-grid quantile inversion, isotonic projection, and uniform bands.}
The fourth experiment evaluates the numerical and inferential components of the uniform QTE theory. 
Unlike the previous experiments, which compare estimators or data-generating regimes, this experiment fixes the transported SA estimator and varies only the outcome grid used to estimate and invert the transported CDFs. We use the same baseline transported missing-outcome design as in Experiment 1 and consider three grids: a coarse fixed grid, a moderate fixed grid, and the growing grid \(J_n=\lceil 4 n^{0.6}\rceil\). 
The QTE curve is evaluated on \(\mathcal T=\{0.10,0.11,\ldots,0.90\}\).

For each grid, we compute the raw one-step CDF estimates, their isotonic projections, the interpolated quantile functions, and simultaneous Wald-type QTE bands based on multiplier draws of the estimated QTE influence function. 
The reported diagnostics are simultaneous coverage over \(\mathcal T\), average band width, \(\sup_{\tau\in\mathcal T}|\widehat\Delta(\tau)-\Delta(\tau)|\), the isotonic projection distance \(d_{\mathrm{iso}}\), and \(\sqrt n d_{\mathrm{iso}}\). 

We also compute an oracle grid-discretization error by applying the same inverse-CDF operation to the true CDF restricted to each numerical grid. 
This last diagnostic isolates the approximation requirement in Proposition~\ref{prop:grid_projection}: fixed grids may look numerically stable in finite samples, but their discretization error is not controlled at the root-\(n\) scale, whereas the growing grid is designed to make the inverse-map error negligible.

\subsection{Results of Simulation Studies}\label{sec:sim_results}
\paragraph{Experiment 1: pointwise finite-sample behavior under transport and surrogate-dependent validation.} Table \ref{tab:exp1_main_n2000} reports the main finite-sample results for Experiment 1.  Across all three quantiles, SA is the best feasible point estimator among the methods compared.  Relative to IPW, SA reduces MSE by 60.0\%, 64.9\%, and 84.9\% at $\tau=0.25,0.50,0.75$, respectively.  Relative to the transported plug-in estimator, the corresponding reductions are 66.1\%, 81.9\%, and 75.3\%.  The source-population negative control is also substantially biased, confirming that the target covariate law is material in this design.  The oracle rows show the performance of the infeasible benchmark and should be read as a reference point rather than as a competing method.

\begin{table}[!htbp]
\centering
\begin{threeparttable}
\caption{Experiment 1: finite-sample performance for transported QTE estimation, $n=2000$}
\label{tab:exp1_main_n2000}
\small
\setlength{\tabcolsep}{5.5pt}
\renewcommand{\arraystretch}{1.12}
\begin{tabular}{cc l r c r r r}
\toprule
$\tau$ & $\Delta_0(\tau)$ & Method & Bias & MSE & RMSE & Coverage & CI length \\
\midrule
0.25 & -3.884 & IPW        & -0.073 & \msegain{3.907}{60.0} & 1.976 & 0.706 & 4.655 \\
0.25 & -3.884 & Oracle     & -0.163 & \mseplain{1.364}      & 1.168 & 0.780 & 3.471 \\
0.25 & -3.884 & Plugin     &  1.961 & \msegain{4.599}{66.1} & 2.144 & \dash & \dash \\
\SArow
0.25 & -3.884 & \textbf{SA} & \textbf{0.023} & \msebest{1.561} & \textbf{1.249} & \textbf{0.770} & \textbf{3.460} \\
0.25 & -3.884 & Source     &  1.757 & \msegain{3.270}{52.3} & 1.808 & \dash & \dash \\
\addlinespace[0.35em]
0.50 & -0.940 & IPW        &  0.037 & \msegain{4.773}{64.9} & 2.185 & 0.754 & 6.376 \\
0.50 & -0.940 & Oracle     & -0.080 & \mseplain{1.053}      & 1.026 & 0.802 & 2.858 \\
0.50 & -0.940 & Plugin     &  2.931 & \msegain{9.275}{81.9} & 3.046 & \dash & \dash \\
\SArow
0.50 & -0.940 & \textbf{SA} & \textbf{0.215} & \msebest{1.675} & \textbf{1.294} & \textbf{0.762} & \textbf{3.521} \\
0.50 & -0.940 & Source     &  2.875 & \msegain{8.483}{80.3} & 2.913 & \dash & \dash \\
\addlinespace[0.35em]
0.75 &  3.150 & IPW        & -0.104 & \msegain{13.398}{84.9} & 3.660 & 0.788 & 9.962 \\
0.75 &  3.150 & Oracle     & -0.030 & \mseplain{0.678}       & 0.824 & 0.816 & 2.289 \\
0.75 &  3.150 & Plugin     &  2.777 & \msegain{8.157}{75.3}  & 2.856 & \dash & \dash \\
\SArow
0.75 &  3.150 & \textbf{SA} & \textbf{0.129} & \msebest{2.017} & \textbf{1.420} & \textbf{0.744} & \textbf{3.453} \\
0.75 &  3.150 & Source     &  2.832 & \msegain{8.173}{75.3} & 2.859 & \dash & \dash \\
\bottomrule
\end{tabular}
\begin{tablenotes}[flushleft]
\footnotesize
\item \textit{Notes.} Results are based on $B_{\mathrm{MC}}=500$ Monte Carlo replications.  $\Delta_0(\tau)$ is the Monte Carlo truth for the transported QTE.  Bias is $\overline{\widehat\Delta}-\Delta_0(\tau)$; MSE is $E\{(\widehat\Delta-\Delta_0)^2\}$; RMSE is $\sqrt{\mathrm{MSE}}$.  Coverage and CI length refer to nominal 95\% Wald intervals based on the implemented influence-function standard error and are reported for one-step estimators.  Plugin and Source are point-estimation baselines, so interval metrics are not reported.  Oracle uses true nuisance functions and is infeasible.  For IPW, Plugin, and Source, the superscript next to MSE reports $100\times(1-\mathrm{MSE}_{\mathrm{SA}}/\mathrm{MSE}_{\mathrm{method}})$; larger values indicate larger MSE reduction by SA.
\end{tablenotes}
\end{threeparttable}
\end{table}

The main pattern is stable across the three parts of the distribution.  IPW is most affected in the upper tail, where the validation-only strategy has much larger variance.  Plugin and Source have smaller Monte Carlo variation but large systematic biases, reflecting approximation error in the transported regression and the failure to target the target covariate law, respectively.  SA combines the low-bias behavior of the one-step correction with materially smaller MSE than the feasible baselines.  Appendix Table \ref{tab:exp1_appendix_n4000} shows that the same ordering persists at $n=4000$, with SA continuing to dominate IPW, Plugin, and Source in MSE at every reported quantile.

\paragraph{Experiment 2: efficiency gains from observing surrogates.}
Table \ref{tab:exp2_efficiency_tau50} reports the main efficiency comparison at the median.  The results follow the pattern predicted by the efficiency-gain formula.  When \(\lambda_S=0\), the NoS/SA MSE ratio is essentially one at every validation rate, indicating that the procedure does not manufacture precision from an uninformative surrogate.  As the surrogate becomes more predictive of the primary outcome, SA becomes substantially more efficient than NoS.  At the lowest validation rate, the NoS/SA MSE ratio increases from 1.00 to 1.64 and 2.32 as \(\lambda_S\) moves from 0 to 1 and 2, corresponding to MSE reductions of 39.2\% and 56.8\%.  The gain is smaller when validation is abundant, as expected from the \((1-\rho_a^0)/\rho_a^0\) factor in Theorem \ref{thm:eff_gain_cdf}.

The empirical MSE ratios closely track the oracle influence-function variance ratios, supporting the interpretation that the observed improvement is the efficiency gain predicted by the theory rather than a nuisance-model artifact. IPW is also less stable than SA, especially under low validation rates, because it relies only on validated primary outcomes.

Appendix Table~\ref{tab:exp2_appendix_ratios} repeats the same efficiency-ratio analysis at $\tau=0.25,0.50,0.75$ and shows that the monotone pattern is not specific to the median. 
To further check that the improvement of SA is not obtained at the cost of bias or unreliable inference, Appendix Figures~\ref{fig:exp2_a3_coverage_length} and~\ref{fig:exp2_a3_bias_mcsd} summarize the full finite-sample diagnostics. These diagnostics show that SA remains nearly unbiased, has broadly comparable coverage to the other one-step estimators, and achieves shorter intervals and smaller empirical dispersion than NoS and IPW.

\begin{table}[!htbp]
\centering
\caption{Experiment 2: efficiency gains from observing surrogates at $\tau=0.50$}
\label{tab:exp2_efficiency_tau50}
\begin{threeparttable}
\small
\setlength{\tabcolsep}{4.6pt}
\renewcommand{\arraystretch}{1.12}
\begin{adjustbox}{max width=\textwidth}
\begin{tabular}{ccrrrrrrr}
\toprule
$\bar\rho$ & $\lambda_S$ & MSE(SA) & MSE(NoS) & NoS/SA MSE & SA gain & NoS/SA length & Theory NoS/SA & IPW/SA MSE \\
\midrule
0.20 & 0.00 & 0.010 & 0.010 & 1.00 & 0.2\% & 0.98 & 1.00 & \textbf{1.18} \\
0.20 & 1.00 & 0.023 & 0.037 & \textbf{1.64} & 39.2\% & 1.25 & \textbf{1.59} & \textbf{1.97} \\
0.20 & 2.00 & 0.044 & 0.101 & \textbf{2.32} & 56.8\% & 1.50 & \textbf{2.19} & \textbf{2.76} \\
0.40 & 0.00 & 0.005 & 0.005 & 1.00 & 0.1\% & 1.00 & 1.00 & \textbf{1.38} \\
0.40 & 1.00 & 0.013 & 0.019 & \textbf{1.48} & 32.3\% & 1.17 & \textbf{1.39} & \textbf{2.08} \\
0.40 & 2.00 & 0.029 & 0.050 & \textbf{1.68} & 40.6\% & 1.31 & \textbf{1.71} & \textbf{2.25} \\
0.70 & 0.00 & 0.003 & 0.003 & 1.00 & -0.1\% & 1.00 & 1.00 & \textbf{1.81} \\
0.70 & 1.00 & 0.008 & 0.009 & \textbf{1.16} & 14.1\% & 1.07 & \textbf{1.15} & \textbf{2.11} \\
0.70 & 2.00 & 0.020 & 0.025 & \textbf{1.29} & 22.5\% & 1.12 & \textbf{1.27} & \textbf{2.25} \\
\bottomrule
\end{tabular}
\end{adjustbox}
\begin{tablenotes}[flushleft]
\footnotesize
\item \textit{Notes.} SA is the proposed surrogate-assisted transported one-step estimator.  NoS is the otherwise analogous transported one-step estimator that does not use the surrogate.  Ratios larger than one favor SA.  SA gain is $100\times(1-\mathrm{MSE}_{\mathrm{SA}}/\mathrm{MSE}_{\mathrm{NoS}})$.  The theory column reports the corresponding oracle influence-function variance ratio under the simulation law.  IPW/SA MSE is included to show the loss from relying only on validated primary outcomes.
\end{tablenotes}
\end{threeparttable}
\end{table}

\paragraph{Experiment 3: source-target covariate shift and density-ratio adjustment.}

The results in Table~\ref{tab:exp3_transport_tau50} follow the pattern predicted by the transport identification argument and the efficient influence function. When \(c=0\), the source and target covariate laws coincide, and the source-population negative control is nearly unbiased for the target QTE. This sanity check confirms that the negative control behaves as expected when no transport adjustment is needed. As \(c\) increases, however, the source-population estimator develops substantial systematic bias because it continues to target the source covariate law rather than the target covariate law. For example, at the median, its bias increases in magnitude from nearly zero at \(c=0\) to about \(-0.70\) at \(c=1.2\) under both randomized and observational assignment.

In contrast, the proposed transported SA estimator remains centered near the target estimand throughout the shift sequence. 
Its MSE and interval length increase moderately as \(c\) grows, reflecting the finite-sample overlap loss induced by more variable density-ratio weights, but the bias remains small. 
Moreover, the feasible SA estimator with estimated \(\widehat\omega\) is nearly indistinguishable from the corresponding oracle-\(\omega\) version. 
This indicates that, in this design, the observed degradation under stronger shift is driven primarily by intrinsic overlap loss rather than by density-ratio estimation error. 
The validation-only IPW estimator is also approximately centered, but it is less efficient, with substantially longer confidence intervals. 
Thus Experiment 3 confirms that the density-ratio component is essential for transported QTE estimation and that the proposed estimator uses it effectively in both randomized and observational source studies.

The density-ratio diagnostics in Appendix Figure~\ref{fig:exp3_omega_diagnostics} show that estimating \(\omega\) becomes more difficult as \(c\) increases, but the MSE of SA using \(\widehat\omega\) remains essentially the same as that of the oracle-\(\omega\) version; hence the finite-sample degradation under stronger shift is primarily an overlap issue rather than a failure of density-ratio estimation.
\begin{table}[!htbp]
\centering
\begin{threeparttable}
\caption{Experiment 3: median transported QTE performance under source-target shift}
\label{tab:exp3_transport_tau50}
\small
\renewcommand{\arraystretch}{1.03}

\begin{tabular}{c l rrrr rrrr}
\toprule
& & \multicolumn{4}{c}{\textit{Randomized}} 
& \multicolumn{4}{c}{\textit{Observational}} \\
\cmidrule(lr){3-6}\cmidrule(lr){7-10}
$c$ & Method 
& Bias & MSE & Coverage & CI length
& Bias & MSE & Coverage & CI length \\
\midrule
\SArow
0.0 & SA                &  0.004 & 0.012 & 0.922 & 0.456 & -0.013 & 0.013 & 0.928 & 0.462 \\
0.0 & SA true $\omega$  &  0.004 & 0.012 & 0.922 & 0.453 & -0.013 & 0.013 & 0.932 & 0.458 \\
0.0 & Source            &  0.004 & 0.006 & --    & --    & -0.004 & 0.006 & --    & --    \\
0.0 & IPW               & -0.016 & 0.021 & 0.994 & 0.869 & -0.036 & 0.019 & 0.996 & 0.954 \\
\addlinespace[0.25em]
\SArow
0.4 & SA                &  0.003 & 0.014 & 0.920 & 0.472 & -0.003 & 0.014 & 0.932 & 0.463 \\
0.4 & SA true $\omega$  &  0.004 & 0.014 & 0.926 & 0.470 & -0.003 & 0.013 & 0.932 & 0.462 \\
0.4 & Source            & -0.254 & 0.070 & --    & --    & -0.254 & 0.070 & --    & --    \\
0.4 & IPW               & -0.018 & 0.020 & 0.982 & 0.790 & -0.027 & 0.016 & 0.994 & 0.830 \\
\addlinespace[0.25em]
\SArow
0.8 & SA                & -0.006 & 0.016 & 0.914 & 0.495 &  0.003 & 0.016 & 0.942 & 0.533 \\
0.8 & SA true $\omega$  & -0.005 & 0.015 & 0.918 & 0.491 &  0.003 & 0.016 & 0.940 & 0.526 \\
0.8 & Source            & -0.498 & 0.254 & --    & --    & -0.488 & 0.244 & --    & --    \\
0.8 & IPW               & -0.028 & 0.026 & 0.970 & 0.776 & -0.018 & 0.018 & 0.988 & 0.826 \\
\addlinespace[0.25em]
\SArow
1.2 & SA                & -0.002 & 0.022 & 0.918 & 0.576 &  0.002 & 0.024 & 0.920 & 0.603 \\
1.2 & SA true $\omega$  & -0.001 & 0.021 & 0.914 & 0.565 &  0.002 & 0.024 & 0.922 & 0.605 \\
1.2 & Source            & -0.702 & 0.498 & --    & --    & -0.697 & 0.492 & --    & --    \\
1.2 & IPW               & -0.015 & 0.031 & 0.962 & 0.862 & -0.014 & 0.027 & 0.970 & 0.880 \\
\bottomrule
\end{tabular}
\begin{tablenotes}[flushleft]
\footnotesize
\item \textit{Notes.} SA uses the estimated density ratio. SA true $\omega$ uses the true density ratio but otherwise estimates the same nuisance functions. Source is the source-population negative control without transport. IPW is a validation-only transported weighting estimator.
\end{tablenotes}
\end{threeparttable}
\end{table}

\paragraph{Experiment 4: growing-grid quantile inversion, isotonic projection, and uniform bands.}

Table~\ref{tab:exp4_uniform_grid} reports the numerical and uniform-in-\(\tau\) diagnostics for Experiment 4. The most direct diagnostic for Proposition~\ref{prop:grid_projection} is the oracle grid error, which isolates inverse-CDF discretization error by applying the same quantile inversion step to the true transported CDF restricted to each numerical grid. For the fixed grids, this error is essentially constant as \(n\) increases, so its root-\(n\) scale contribution grows with sample size. In contrast, the growing grid makes the oracle grid error negligible at the root-\(n\) scale. This supports the theoretical requirement that the numerical grid be refined fast enough for the quantile inversion step not to contaminate the first-order expansion.

The remaining columns summarize the finite-sample behavior of the full SA estimator and the multiplier simultaneous bands. The sup-norm QTE error decreases with \(n\) across all grid choices, and the simultaneous coverage over \(\mathcal T=\{0.10,\ldots,0.90\}\) remains at or above the nominal level in these designs. The growing grid is conservative in finite samples, as reflected by its wider simultaneous bands, but it is also the only grid whose deterministic inversion error is controlled at the root-\(n\) scale. Thus the experiment illustrates the distinction between finite-sample smoothing effects from coarse grids and the asymptotic grid-refinement condition required by the uniform quantile theory. Appendix Figure~\ref{fig:exp4_width_iso} further reports the band-width and isotonic-projection diagnostics underlying these results.

\begin{table}[!htbp]
\centering
\begin{threeparttable}
\caption{Experiment 4: grid design and simultaneous QTE bands}
\label{tab:exp4_uniform_grid}
\small
\renewcommand{\arraystretch}{1.00}

\begin{tabular*}{\textwidth}{@{\extracolsep{\fill}}lrrrrrr@{}}
\toprule
Grid & $n$ & $J$ & Oracle grid err. & Sup. error & Unif. cover & $\sqrt n d_{\rm iso}$ \\
\midrule
fixed101 & 2000 & 101 & 0.0010 & 0.347 & 0.998 & 0.160 \\
fixed51  & 2000 & 51  & 0.0054 & 0.332 & 0.996 & 0.145 \\
growing  & 2000 & 383 & 0.0001 & 0.361 & 1.000 & 0.182 \\
fixed101 & 4000 & 101 & 0.0010 & 0.243 & 0.996 & 0.085 \\
fixed51  & 4000 & 51  & 0.0054 & 0.230 & 0.988 & 0.075 \\
growing  & 4000 & 580 & 0.0000 & 0.255 & 1.000 & 0.099 \\
fixed101 & 8000 & 101 & 0.0010 & 0.171 & 0.974 & 0.050 \\
fixed51  & 8000 & 51  & 0.0054 & 0.162 & 0.974 & 0.044 \\
growing  & 8000 & 879 & 0.0000 & 0.178 & 1.000 & 0.058 \\
\bottomrule
\end{tabular*}

\begin{tablenotes}[flushleft]
\footnotesize
\item \textit{Notes.} Uniform coverage is simultaneous over $\mathcal{T}=\{0.10,0.11,\ldots,0.90\}$ using Wald-type multiplier bands. Oracle grid err. is the deterministic inverse-CDF discretization error obtained by restricting the true CDF to the reported grid. Sup. error is $\sup_{\tau\in\mathcal{T}}|\widehat{\Delta}(\tau)-\Delta(\tau)|$. $d_{\rm iso}$ is the maximum projection distance between raw and isotonic CDF estimates over treatment arms and grid points.
\end{tablenotes}
\end{threeparttable}
\end{table}

\section{Real Data Illustration}\label{sec:realdata}

We illustrate the proposed method using the ACTG 175 HIV clinical trial \citep{hammer1996trial,actg175data}. 
The trial enrolled adults with HIV infection and baseline CD4 counts between 200 and 500 cells per cubic millimeter, and compared zidovudine monotherapy with alternative antiretroviral regimens. 
For the present analysis, the primary outcome \(Y\) is CD4 count at 96 weeks, denoted by \texttt{cd496}; the validation indicator \(M\) is \texttt{r}, which records whether the 96-week CD4 measurement is observed. 
The surrogate vector \(S\) consists of 20-week CD4 and CD8 counts, \texttt{cd420} and \texttt{cd820}. 
Baseline covariates \(X\) include demographic, clinical, treatment-history, and baseline immunologic variables, including baseline CD4 and CD8 counts. 
Treatment \(A\) is coded as non-zidovudine-only therapy versus zidovudine-only therapy.

ACTG 175 does not contain an external target-only cohort. 
We therefore use it as an internal transport illustration. 
A fixed source/target split is constructed using only baseline covariates, and the estimators use only \(X\) from the target units. 
All target post-baseline variables, including \(A\), \(S\), \(M\), and \(Y\), are held out from estimation. 
This construction gives observed data of the same form as in the paper,
\[
O=(R,X,R A,R S,R M,R M Y),
\]
while allowing a real-data assessment of the transported QTE estimator under missing primary outcomes. 
The purpose of the analysis is not to validate a known ground truth, but to examine whether early surrogate biomarkers improve stability relative to no-surrogate and validation-only alternatives in a real clinical-trial setting.

The constructed sample contains 1,408 source units and 731 target units. 
The source validation rate is 0.627, so a substantial fraction of source primary outcomes are missing. 
The density-ratio effective sample size is 682.9, equal to 48.5\% of source units, indicating nontrivial but usable source-target overlap. 
The early biomarkers are predictive of the primary outcome: in the source validation sample, adding 20-week CD4 and CD8 to the baseline covariates increases the predictive \(R^2\) for 96-week CD4 from 0.313 to 0.504. 
This provides empirical motivation for using the surrogate-assisted estimator.

\begin{table}[!htbp]
\centering
\begin{threeparttable}
\caption{ACTG 175 empirical illustration: transported QTE estimates for 96-week CD4 count}
\label{tab:actg175_qte}
\small
\renewcommand{\arraystretch}{1.05}

\begin{tabular*}{\textwidth}{@{\extracolsep{\fill}}ccrrrr@{}}
\toprule
$\tau$ & Method & $\widehat\Delta(\tau)$ & SE & 95\% CI & CI length ratio \\
\midrule
\textbf{0.25} & \textbf{SA} & \textbf{121.1} & \textbf{2.8} & \textbf{[115.6, 126.6]} & \textbf{1.00} \\
0.25 & NoS & 111.8 & 4.8 & [102.4, 121.1] & 1.70 \\
0.25 & IPW & 113.4 & 3.1 & [107.3, 119.6] & 1.12 \\
0.25 & Source & 112.8 & -- & -- & -- \\
\textbf{0.50} & \textbf{SA} & \textbf{78.0} & \textbf{5.0} & \textbf{[68.1, 87.9]} & \textbf{1.00} \\
0.50 & NoS & 72.3 & 5.0 & [62.5, 82.1] & 0.99 \\
0.50 & IPW & 64.0 & 5.7 & [52.9, 75.2] & 1.12 \\
0.50 & Source & 51.5 & -- & -- & -- \\
\textbf{0.75} & \textbf{SA} & \textbf{112.0} & \textbf{5.4} & \textbf{[101.4, 122.7]} & \textbf{1.00} \\
0.75 & NoS & 105.8 & 5.4 & [95.2, 116.4] & 0.99 \\
0.75 & IPW & 72.4 & 6.9 & [58.8, 86.0] & 1.28 \\
0.75 & Source & 55.2 & -- & -- & -- \\
\bottomrule
\end{tabular*}

\begin{tablenotes}[flushleft]
\footnotesize
\item \textit{Notes.} The target covariate distribution is constructed by a fixed baseline-only internal transport split. Only target baseline covariates are used by the estimators. The primary outcome is 96-week CD4 count; surrogates are 20-week CD4 and CD8 counts. Source is a negative-control source-population estimate. The source validation rate is 0.627, and the density-ratio ESS is 682.9 (48.5\% of source units).
\end{tablenotes}
\end{threeparttable}
\end{table}

\begin{figure}[!htbp]
    \centering
    \includegraphics[width=0.92\linewidth]{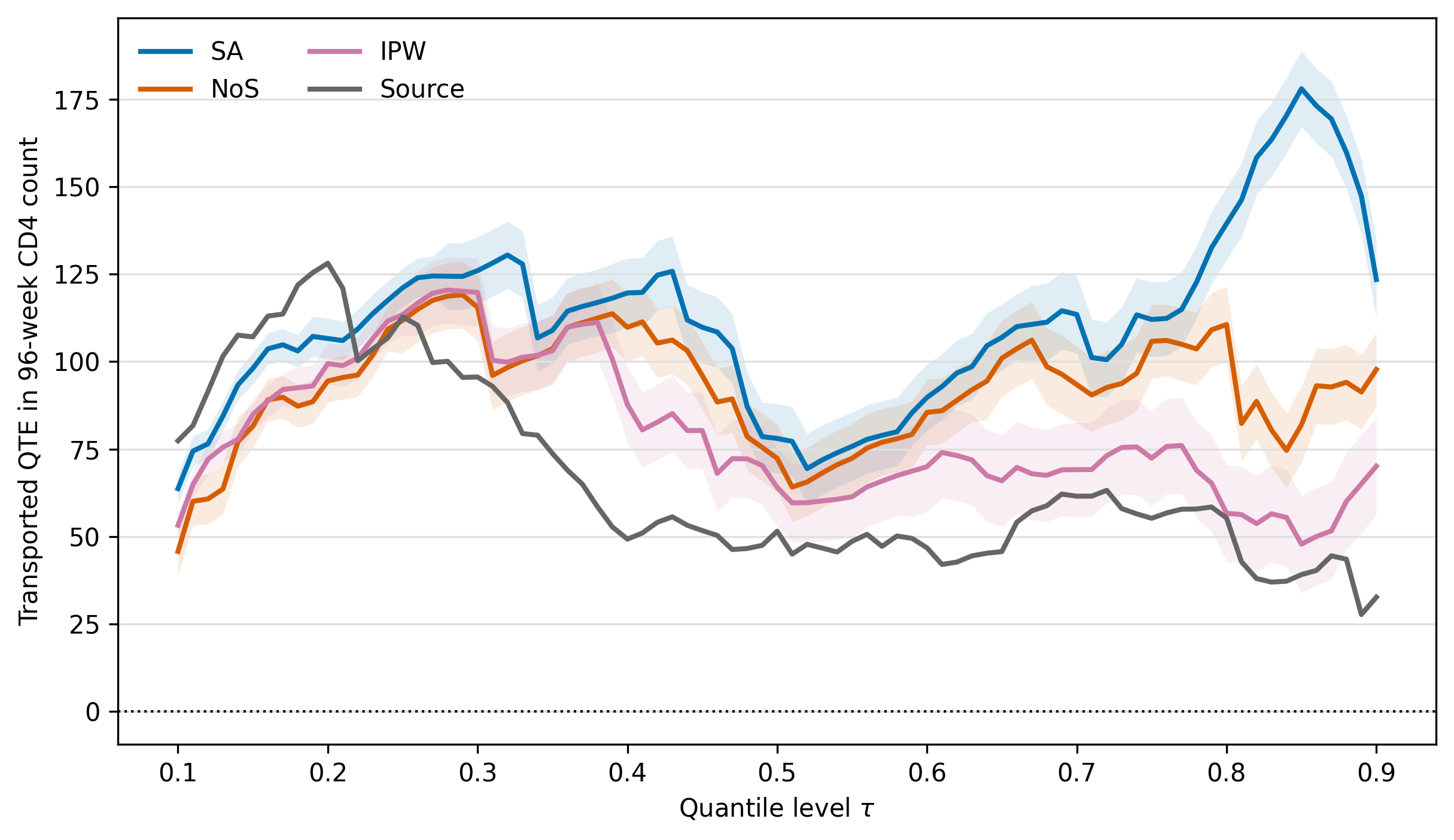}
    \caption{
    ACTG 175 empirical illustration: transported QTE curve for 96-week CD4 count. 
    The figure reports \(\widehat\Delta(\tau)\) over \(\tau\in[0.10,0.90]\). 
    Shaded bands are pointwise 95\% Wald intervals for SA, NoS, and IPW. 
    Source is a source-population negative-control curve and does not target the transported estimand. 
    The SA curve is positive over the central quantile range, while the validation-only IPW curve is less stable and the source-population curve differs materially from the transported curves.
    }
    \label{fig:actg175_qte_curve}
\end{figure}

Table~\ref{tab:actg175_qte} and Figure~\ref{fig:actg175_qte_curve} show positive transported QTE estimates for 96-week CD4 count across the reported quantiles. 
At \(\tau=0.25,0.50,0.75\), the SA estimates are 121.1, 78.0, and 112.0 CD4 cells, respectively. 
The comparison with IPW illustrates the value of using surrogate information: the validation-only IPW intervals are wider at all three reported quantiles, with CI length ratios relative to SA equal to 1.12, 1.12, and 1.28. The comparison with NoS is more localized. SA gives a substantially shorter interval at the lower quartile, while the median and upper-quartile intervals are similar. This is consistent with the surrogate providing useful information about parts of the outcome distribution, without implying uniform dominance at every quantile in a single real-data analysis. 
Finally, the source-population curve and estimates are notably smaller at the median and upper quartile, highlighting that the target-population transported estimand can differ materially from the source-population estimand.

Overall, the empirical illustration supports the main message of the paper: when early biomarkers are predictive and the long-term primary outcome is incompletely observed, the proposed estimator can use the surrogate process to improve stability relative to validation-only weighting while targeting the desired transported distributional estimand. Additional implementation diagnostics, including interval-length ratios and source-target balance before and after density-ratio weighting, are reported in Appendix~\ref{app:actg175_details}.

\section{Discussion and Scope}\label{sec:discussion}

The method uses the surrogate as auxiliary information rather than as a replacement endpoint. It does not require the surrogate to mediate the treatment effect. The identifying assumptions are missing at random for the primary outcome and conditional transportability of the treatment-specific surrogate-primary outcome process. The core theoretical contribution is the three-layer canonical gradient, together with the exact drift identity and the quantile-specific efficiency-gain formula. The computational contribution is that, after nuisance fitting, the estimator is a closed-form one-step update and the recommended nuisance subproblems can be implemented with convex or normal-equation procedures. Taken together, the simulations and ACTG 175 illustration indicate that the proposed framework is practically implementable and that its main components have distinct empirical roles: surrogate measurements improve efficiency, density-ratio weighting targets the desired covariate law, and growing-grid quantile inversion stabilizes uniform inference. Future work may extend the empirical component to settings with external target cohorts and richer high-dimensional surrogate processes.

\clearpage
\appendix
\setcounter{table}{0}
\renewcommand{\thetable}{A.\arabic{table}}
\renewcommand{\theHtable}{A.\arabic{table}}
\setcounter{figure}{0}
\renewcommand{\thefigure}{A.\arabic{figure}}
\renewcommand{\theHfigure}{A.\arabic{figure}}

\section{Proofs for Identification and Efficient Influence Functions}

\begin{proof}[Proof of Proposition \ref{prop:identification}]
Fix $a$ and $y$. By the definition of the target estimand and iterated expectation,
\[
\psi_a(y)=E\{P(Y^a\le y\mid X,R=0)\mid R=0\}.
\]
By Assumption \ref{ass:transport},
\[
P(Y^a\le y\mid X,R=0)=P(Y^a\le y\mid X,R=1).
\]
Conditioning on $S^a$ gives
\[
P(Y^a\le y\mid X,R=1)=E\{P(Y^a\le y\mid X,S^a,R=1)\mid X,R=1\}.
\]
By Assumption \ref{ass:causal}, $(S^a,Y^a)\perp A\mid X,R=1$, and by consistency, among source units with $A=a$, $(S,Y)=(S^a,Y^a)$. Hence
\[
S^a\mid X,R=1\overset d=S\mid X,A=a,R=1,
\]
and
\[
P(Y^a\le y\mid X,S^a=s,R=1)=P(Y\le y\mid X,S=s,A=a,R=1).
\]
By Assumption \ref{ass:mar},
\begin{align*}
P(Y\le y\mid X,S=s,A=a,R=1)
&=P(Y\le y\mid X,S=s,A=a,R=1,M=1)\\
&=m_a(y,X,s).
\end{align*}
Combining these displays yields
\[
P(Y^a\le y\mid X,R=0)=E\{m_a(y,X,S)\mid X,A=a,R=1\}=g_a(y,X).
\]
Taking expectation over $X\mid R=0$ proves $\psi_a(y)=E_0\{g_a(y,X)\}$. The unconditional representation with $\ind(R=0)/\pi_0$ is immediate. The source representation follows from $dP_{X,0}=\omega dP_{X,1}$.
\end{proof}

\begin{proof}[Proof of Theorem \ref{thm:eif}]
Fix $a$ and $y$ and suppress these indices where doing so is harmless. The proof is for the observed-data functional $\Psi_a^y(P)=E_P\{g_P(y,X)\mid R=0\}$; Proposition \ref{prop:identification} then maps it to the causal estimand. The argument is the usual observed-data tangent-space calculation for a pathwise differentiable functional \citep{Newey1994,RobinsRotnitzky1995,Tsiatis2006}. Consider any regular parametric submodel $P_\varepsilon$ through $P$ with score $\ell(O)$ at $\varepsilon=0$. The observed-data score has the orthogonal decomposition
\[
\ell=\ell_R+\ind(R=0)\ell_{X,0}+\ind(R=1)\{\ell_{X,1}+\ell_A+\ell_S+\ell_M+M\ell_Y\},
\]
where each conditional score has mean zero given its conditioning variables. Thus the observed-data tangent space is the orthogonal direct sum of the closed linear spans generated by these score components.

The pathwise derivative of
\[
\Psi_\varepsilon=E_{0,\varepsilon}\{g_\varepsilon(X)\},\qquad
 g_\varepsilon(x)=E_\varepsilon\{m_\varepsilon(X,S)\mid R=1,A=a,X=x\},
\]
with $m_\varepsilon(x,s)=E_\varepsilon(Z_y\mid R=1,A=a,X=x,S=s,M=1)$ has three nonzero components. First, perturbing $P_{X\mid R=0}$ gives
\[
\dot\Psi_{X0}=E\left[\frac{\ind(R=0)}{\pi_0}\{g(X)-\Psi\}\ell(O)\right].
\]
The score $\ell_R$ for the marginal law of $R$ contributes zero because $\Psi$ is conditional on $R=0$, and the centered form above is orthogonal to $\ell_R$.

Second, perturbing the conditional law $P_{S\mid R=1,A=a,X}$ gives
\begin{align*}
\dot\Psi_S
&=E_0\left[E\{(m(X,S)-g(X))\ell_S\mid R=1,A=a,X\}\right]\\
&=E\left[\frac{\ind(R=1)\omega(X)\ind(A=a)}{\pi_1e_a(X)}\{m(X,S)-g(X)\}\ell(O)\right],
\end{align*}
where the second equality follows by changing measure from $X\mid R=0$ to $X\mid R=1$ using $\omega$ and by reweighting the event $A=a$ within the source. The representing function lies in the closed tangent subspace for $S\mid R=1,A,X$ and is orthogonal to the source $X$ and treatment tangent spaces because its conditional mean given $(R=1,A,X)$ is zero.

Third, perturbing $P_{Y\mid R=1,A=a,X,S,M=1}$ gives
\begin{align*}
\dot\Psi_Y
&=E_0\left[E\{E[(Z_y-m(X,S))\ell_Y\mid R=1,A=a,X,S,M=1]\mid R=1,A=a,X\}\right]\\
&=E\left[\frac{\ind(R=1)\omega(X)\ind(A=a)M}{\pi_1e_a(X)\rho_a(X,S)}\{Z_y-m(X,S)\}\ell(O)\right].
\end{align*}
This representing function belongs to the closed tangent subspace for the validated primary-outcome law and is orthogonal to the missingness tangent space because its conditional mean given $(R=1,A,X,S,M)$ is zero.

The remaining tangent components have zero derivative and zero projection. Perturbing $P_{X\mid R=1}$ alone changes neither the target conditional law $P_{X\mid R=0}$ nor the conditional regressions entering $g$; its possible appearance through the Radon--Nikodym derivative is already only a representation of the change of measure and cancels in the pathwise derivative. Perturbing $P_{A\mid X,R=1}$ or $P_{M\mid A,X,S,R=1}$ changes neither $g(x)$ nor $m(x,s)$ as conditional distributions; the displayed representer is orthogonal to these score spaces because the surrogate and outcome residuals have the corresponding conditional means zero.

Adding the three derivative representers yields $E\{\phi_a^y(O)\ell(O)\}$ for every score $\ell$ in the nonparametric observed-data tangent space. The three summands are mutually orthogonal closed-subspace components, and $E\phi_a^y=0$ by iterated expectation. Hence $\phi_a^y$ is the canonical gradient and its variance is the efficiency bound. In a submodel in which $e_a$ is known, the tangent component for $A\mid X,R=1$ is removed; in a submodel in which $\rho_a$ is known, the tangent component for $M\mid A,X,S,R=1$ is removed. Since $\phi_a^y$ is already orthogonal to both components, its projection onto the smaller tangent space is unchanged. This proves the invariance of the canonical gradient and bound under known treatment or missingness mechanisms.
\end{proof}

\begin{proof}[Proof of Corollary \ref{cor:mean_special_case}]
The proof is identical to the proof of Theorem \ref{thm:eif} with the threshold outcome $Z_y=\ind(Y\le y)$ replaced by the integrable primary outcome $Y$. The pathwise derivative through the target covariate law gives the first term, the derivative through the source surrogate law gives the second term, and the derivative through the validated primary-outcome law gives the third term. These three terms lie in mutually orthogonal tangent components and are mean zero, so their sum is the canonical gradient. Taking the difference between $a=1$ and $a=0$ gives the transported ATE gradient.

When $P_{X,0}=P_{X,1}$, $\omega=1$, so the two source residual terms match the surrogate-process and validation-outcome residuals appearing in the surrogate-assisted ATE gradient of \citet{KallusMao2024Surrogates}, with $M$ playing the role of their primary-outcome observation indicator. The target-covariate term in the present formula is still generated by the two-sample transport design. If that design is further collapsed to a single-population model in which the covariate law is observed from the same population as the source study, the first term becomes the usual baseline regression component and the displayed gradient is exactly the armwise surrogate-assisted ATE canonical gradient in the reference model. This proves both the transported mean result and the stated precise relation to the reference paper.
\end{proof}

\begin{proof}[Proof of Corollary \ref{cor:eif_quantile}]
For fixed $a$, the quantile functional is the inverse map $q_a(\tau)=\psi_a^{-1}(\tau)$. This is the same inverse-map step used in semiparametric quantile and QTE inference \citep{KoenkerBassett1978,Firpo2007,ChernozhukovFernandezValMelly2013}. Under Assumption \ref{ass:quantile}, the inverse map is Hadamard differentiable tangentially to continuous functions \citep{vdVaart1998}, with derivative
\[
h\mapsto -\frac{h(q_a(\tau))}{f_a(q_a(\tau))}.
\]
Applying this derivative to the CDF influence function in Theorem \ref{thm:eif} gives the displayed influence function for $q_a(\tau)$. The QTE is the difference between the two treatment-specific quantiles, so linearity yields the displayed influence function for $\Delta(\tau)$.
\end{proof}

\begin{proof}[Proof of Proposition \ref{prop:drift}]
Fix $a,y$. By iterated expectation and Definition \ref{def:signal},
\begin{align*}
\Psi_a^y(\bar\eta)
&=E_0\{\bar g_a(y,X)\}\\
&\quad+E_1\left[\bar w_a(X)H_a\{\bar m_a(y,X,S)-\bar g_a(y,X)\}(X)\right]\\
&\quad+E_1\left[\bar w_a(X)H_a\{\bar r_a(X,S)[m_a(y,X,S)-\bar m_a(y,X,S)]\}(X)\right].
\end{align*}
Also $\psi_a(y)=E_0\{g_a(y,X)\}=E_1\{\omega(X)g_a(y,X)\}$ and $g_a=H_am_a$. Let $\Delta_g=\bar g_a-g_a$ and $\Delta_m=\bar m_a-m_a$. Then
\[
H_a(\bar m_a-\bar g_a)=H_a\Delta_m-\Delta_g,\qquad
H_a\{\bar r_a(m_a-\bar m_a)\}=-H_a(\bar r_a\Delta_m).
\]
Substitution yields
\begin{align*}
\Psi_a^y(\bar\eta)-\psi_a(y)
&=E_1\{\omega(X)\Delta_g(y,X)\}
+E_1\{\bar w_a(X)[H_a\Delta_m(X)-\Delta_g(y,X)-H_a(\bar r_a\Delta_m)(X)]\}\\
&=E_1\{[\omega(X)-\bar w_a(X)]\Delta_g(y,X)\}
+E_1\{\bar w_a(X)H_a[(1-\bar r_a)\Delta_m](X)\},
\end{align*}
which proves \eqref{eq:drift}. The two robustness statements are immediate. The norm bound follows by Cauchy--Schwarz, Jensen's inequality for $H_a$, and boundedness of $\bar w_a$.
\end{proof}

\section{Proofs for Estimation and Inference}

\begin{lemma}[Random-denominator empirical expansion]\label{lem:ratio}
Let $F$ be square-integrable with $P_rF=E(F\mid R=r)$. Then
\[
\Pnr{r}F-P_rF
=\Pn\left[\frac{\ind(R=r)}{\pi_r}\{F-P_rF\}\right]+o_p(n^{-1/2}).
\]
If $F=\widehat F$ is cross-fitted and $\norm{\widehat F-F}_{P,2}=o_p(1)$, the same expansion holds with $F$ replaced by $\widehat F$ up to an additional $o_p(n^{-1/2})$ term.
\end{lemma}

\begin{proof}[Proof of Lemma \ref{lem:ratio}]
Write $\hat\pi_r=\Pn\ind(R=r)$. Then $\Pnr{r}F=\Pn\{\ind(R=r)F\}/\hat\pi_r$. A first-order Taylor expansion of $(u,v)\mapsto u/v$ at $(\pi_rP_rF,\pi_r)$ gives
\begin{align*}
\Pnr{r}F-P_rF
&=\frac{1}{\pi_r}\Pn\{\ind(R=r)F\}-\frac{P_rF}{\pi_r}\Pn\{\ind(R=r)\}+o_p(n^{-1/2})\\
&=\Pn\left[\frac{\ind(R=r)}{\pi_r}\{F-P_rF\}\right]+o_p(n^{-1/2}).
\end{align*}
For cross-fitted $\widehat F$, decompose $\widehat F=F+(\widehat F-F)$. Conditional on the training folds, Chebyshev's inequality gives
\[
\Pn\left[\frac{\ind(R=r)}{\pi_r}\{(\widehat F-F)-P_r(\widehat F-F)\}\right]
=O_p\{n^{-1/2}\norm{\widehat F-F}_{P,2}\}=o_p(n^{-1/2}).
\]
The Taylor remainder is unchanged because $\hat\pi_r-\pi_r=O_p(n^{-1/2})$ and $\norm{\widehat F-F}_{P,2}=o_p(1)$.
\end{proof}

\begin{lemma}[Uniform random-denominator expansion]\label{lem:uniform_ratio}
Let $\widehat{\mathcal F}_k=\{\widehat F_y:y\in\cY\}$ be a fold-specific class independent of the evaluation fold conditional on the training data. Suppose that, with probability tending to one, it is pointwise measurable, has envelope $\widehat F_{e,k}$ with $\|\widehat F_{e,k}\|_{\infty}=O_p(1)$, and satisfies
\[
\sup_Q\log N\{\epsilon\|\widehat F_{e,k}\|_{Q,2},\widehat{\mathcal F}_k,L_2(Q)\}
\le V_n\log(A_n/\epsilon),\qquad 0<\epsilon\le1,
\]
for deterministic $V_n\ge1$ and $A_n\ge e$, with
\[
\frac{V_n\log A_n}{\sqrt n}=O(1).
\]
Then, for $r\in\{0,1\}$,
\[
\sup_{y\in\cY}\left|
\Pnr{r}\widehat F_y-P_r\widehat F_y
-\Pn\left[\frac{\ind(R=r)}{\pi_r}\{\widehat F_y-P_r\widehat F_y\}\right]
\right|=o_p(n^{-1/2})
\]
provided either
\[
\sup_{y\in\cY}\left|\Pn\left[\ind(R=r)\{\widehat F_y-P_r\widehat F_y\}\right]\right|=o_p(1),
\]
or, more primitively,
\[
\sup_{y\in\cY}|P_r\widehat F_y|=O_p(1)
\quad\text{and}\quad
\frac{V_n\log A_n}{n}=o(1).
\]
The same statement holds for finite sums of such fold-specific classes over a fixed number of cross-fitting folds. The bounded-envelope condition may be replaced by an appropriate higher-moment tail condition together with the corresponding maximal inequality.
\end{lemma}

\begin{proof}[Proof of Lemma \ref{lem:uniform_ratio}]
Write $\hat\pi_r=\Pn\ind(R=r)$, $\mu_y=P_r\widehat F_y$, and
\[
A_y=\Pn\{\ind(R=r)(\widehat F_y-\mu_y)\}.
\]
Then $\Pn\{\ind(R=r)\widehat F_y\}=A_y+\mu_y\hat\pi_r$, and hence the ratio identity is exact:
\begin{align*}
\Pnr{r}\widehat F_y-P_r\widehat F_y
-\Pn\left[\frac{\ind(R=r)}{\pi_r}\{\widehat F_y-P_r\widehat F_y\}\right]
&=A_y\left(\frac1{\hat\pi_r}-\frac1{\pi_r}\right).
\end{align*}
Since $\hat\pi_r-\pi_r=O_p(n^{-1/2})$ and $\pi_r$ is bounded away from zero, it suffices to prove $\sup_y|A_y|=o_p(1)$. This is assumed in the first route. Under the primitive entropy route, conditional on the training data, the centered class $\{\ind(R=r)(\widehat F_y-P_r\widehat F_y):y\in\cY\}$ has the same entropy order and a bounded envelope with probability tending to one. The maximal inequality for VC-type classes gives
\[
E\left[\sup_y |(\Pn-P)\{\ind(R=r)(\widehat F_y-P_r\widehat F_y)\}|\mid\text{training}\right]
\lesssim \sqrt{\frac{V_n\log A_n}{n}}+\frac{V_n\log A_n}{n},
\]
up to envelope constants, which is $o_p(1)$ when $V_n\log A_n=o(n)$. The centered class has mean zero by definition, so the left-hand side controls $\sup_y|A_y|$. Therefore $\sup_y|A_y|=o_p(1)$ and the displayed expansion follows. Averaging over a fixed number of folds preserves the order.
\end{proof}

\begin{proof}[Proof of Proposition \ref{prop:closed_convex}]
Given nuisance estimates, \eqref{eq:estimator} is a finite sum of observed quantities. For the implementation class in Definition \ref{def:convex}, and for the finite-dimensional version in Definition \ref{def:sieve_learners}, ridge least squares and finite-rank KRR solve normal equations; ridge logistic regression minimizes a convex negative log-likelihood plus a convex quadratic penalty; and entropy balancing under the stated exponential-tilt calibration model has a convex entropy objective with linear constraints. The isotonic step is Euclidean projection onto a closed convex cone intersected with a box, and quantile inversion over a finite grid is an order operation. Thus the recommended implementation uses only closed-form or convex subproblems. This proposition does not assert convexity or rate validity for arbitrary flexible learners outside Definition \ref{def:sieve_learners}.
\end{proof}

\begin{proof}[Proof of Theorem \ref{thm:fixed_al}]
Fix $a,y$ and suppress them from notation. Let
\[
B(O)=\frac{\omega(X)\ind(A=a)}{e_a(X)}\{m(y,X,S)-g(y,X)\}
+\frac{\omega(X)\ind(A=a)M}{e_a(X)\rho_a(X,S)}\{Z_y-m(y,X,S)\},
\]
so that $E_1B=0$. Let $\widehat B$ be the same expression with estimated nuisances. The estimator is $\widehat\psi=\Pnr{0}\widehat g+\Pnr{1}\widehat B$. Add and subtract oracle terms:
\[
\widehat\psi-\psi=(\Pnr{0}g-E_0g)+\Pnr{1}B+D_n+S_n,
\]
where $D_n=E_0\widehat g+E_1\widehat B-\psi$ and
\[
S_n=(\Pnr{0}-E_0)(\widehat g-g)+(\Pnr{1}-E_1)(\widehat B-B).
\]
Proposition \ref{prop:drift} and \eqref{eq:productrate} imply $D_n=o_p(n^{-1/2})$. By cross-fitting, conditional on training folds, the evaluation-fold empirical processes are independent of nuisance estimates. Conditional Chebyshev's inequality and the $L_2$ signal convergence in Assumption \ref{ass:rates_fixed} give $S_n=o_p(n^{-1/2})$, with Lemma \ref{lem:ratio} handling random denominators. Therefore
\[
\widehat\psi-\psi=(\Pnr{0}g-E_0g)+\Pnr{1}B+o_p(n^{-1/2}).
\]
Applying Lemma \ref{lem:ratio} to the first term and to the source term, noting $E_1B=0$, gives
\[
\widehat\psi-\psi
=\Pn\left[\frac{\ind(R=0)}{\pi_0}\{g(X)-\psi\}
+\frac{\ind(R=1)}{\pi_1}B(O)\right]+o_p(n^{-1/2}).
\]
The bracketed term is exactly $\phi_a^y(O)$. Multiplication by $\sqrt n$ gives the expansion. The central limit theorem applies because the influence function is square-integrable under positivity and boundedness. Efficiency follows from Theorem \ref{thm:eif}.
\end{proof}

\begin{proof}[Proof of Corollary \ref{cor:pointwise}]
Fix $a$ and write $q=q_a(\tau)$, $\widehat q=\widehat q_a(\tau)$, and $f=f_a(q)$. The local uniform consistency in Assumption \ref{ass:local_quantile} and the positivity of $f$ imply $\widehat q-q=o_p(1)$ by the usual bracketing argument: for each fixed $\varepsilon>0$ small enough that $q\pm\varepsilon\in\mathcal N_a$, differentiability gives $\psi_a(q-\varepsilon)<\tau<\psi_a(q+\varepsilon)$, and the same inequalities hold with $\psi_a$ replaced by $\widehat F_a$ with probability tending to one.

It remains to identify the first-order term. Since $\Pn\phi_a^q=O_p(n^{-1/2})$, for fixed $t>0$ define
\[
y_{n,\pm}=q-\frac{\Pn\phi_a^q}{f}\pm \frac{t}{\sqrt n}.
\]
These points lie in $\mathcal N_a$ with probability tending to one. By differentiability of $\psi_a$ at $q$ and by the local empirical-process equicontinuity in Assumption \ref{ass:local_quantile},
\begin{align*}
\widehat F_a(y_{n,\pm})-\tau
&=\psi_a(y_{n,\pm})-\psi_a(q)+\Pn\phi_a^{y_{n,\pm}}+o_p(n^{-1/2})\\
&=f(y_{n,\pm}-q)+\Pn\phi_a^q+o_p(n^{-1/2})\\
&=\pm \frac{ft}{\sqrt n}+o_p(n^{-1/2}).
\end{align*}
Hence, with probability tending to one, $\widehat F_a(y_{n,-})\le\tau\le\widehat F_a(y_{n,+})$ for all sufficiently large $t$ after an arbitrarily small enlargement of $t$. Monotonicity of $\widehat F_a$ and the generalized-inverse definition therefore give
\[
\widehat q=q-\frac{\Pn\phi_a^q}{f}+o_p(n^{-1/2}).
\]
Multiplying by $\sqrt n$ gives the first assertion. Subtracting the representations for $a=1$ and $a=0$ gives the QTE representation. The central limit theorem applies to the square-integrable influence function $\phi_{\Delta,\tau}$.
\end{proof}

\begin{proof}[Proof of Lemma \ref{lem:oracle_donsker}]
Assumption \ref{ass:uniform_oracle}(iii) states that $\Phi_a$ is pointwise measurable and $P$-Donsker with a bounded envelope. The functional central limit theorem for Donsker classes therefore gives the displayed convergence. The covariance formula is the covariance of the limiting isonormal process applied to the class $\Phi_a$. Tightness and sample-path continuity with respect to the covariance semimetric follow from the Donsker property and the assumed covariance continuity.
\end{proof}

\begin{proof}[Proof of Lemma \ref{lem:small_ep}]
Condition on the training folds. The class is then fixed relative to the evaluation fold. A standard maximal inequality for classes with entropy $\log N(\epsilon\|G\|_{Q,2},\mathcal G,L_2(Q))\le V_n\log(A_n/\epsilon)$ and $L_2(P)$ radius $\delta_n$, as in empirical-process treatments of VC-type classes and growing-complexity suprema \citep{vdVW1996,ChernozhukovChetverikovKato2014}, gives
\[
E\left[\sup_{f\in\widehat{\mathcal G}_k}\left|\sqrt n(\mathbb P_{n,k}-P)f\right|\mid\text{training}\right]
\lesssim
\delta_n\sqrt{V_n\log(A_n/\delta_n)}+
\frac{V_n\log(A_n/\delta_n)}{\sqrt n},
\]
up to constants depending only on the envelope bound. The displayed entropy-growth condition makes the right-hand side $o(1)$. Markov's inequality yields
\[
\sup_{f\in\widehat{\mathcal G}_k}\left|\sqrt n(\mathbb P_{n,k}-P)f\right|=o_p(1),
\]
and division by $\sqrt n$ proves the claim. Averaging over a fixed number of folds preserves the order.
\end{proof}

\begin{proof}[Proof of Proposition \ref{prop:uniform_linearization}]
Fix $a$. Write the estimator as a fold average. For a generic fold, decompose uniformly in $y$:
\[
\widehat\psi_a(y)-\psi_a(y)
=(\Pn-P)\phi_a^y+D_n(y)+E_n(y)+o_p(n^{-1/2}),
\]
where the random-denominator remainder is $o_p(n^{-1/2})$ uniformly by Lemma \ref{lem:uniform_ratio},
\[
D_n(y)=E\{\widehat\Gamma_{a,k}^y(O)\}-E\{\Gamma_a^y(O)\}
\]
is the drift, and $E_n(y)=(\Pn-P)\{\widehat\Gamma_{a,k}^y-\Gamma_a^y\}$ is the empirical-process perturbation. Proposition \ref{prop:drift} gives, uniformly over $y\in\cY$,
\[
|D_n(y)|\le C r_{w,a,n}r_{g,a,n}+C r_{\rho\mathrm{rat},a,n}r_{m,a,n}=o_p(n^{-1/2}).
\]
The class of perturbations satisfies Lemma \ref{lem:small_ep} by Assumption \ref{ass:uniform_nuisance}, so $\sup_{y\in\cY}|E_n(y)|=o_p(n^{-1/2})$. Combining these bounds over the fixed number of folds proves the uniform expansion.
\end{proof}

\begin{proof}[Proof of Theorem \ref{thm:uniform}]
By Proposition \ref{prop:uniform_linearization},
\[
\sqrt n\{\widehat\psi_a-\psi_a\}=\sqrt n(\Pn-P)\phi_a^{\cdot}+o_p(1)
\]
in $\ell^\infty(\cY)$. Lemma \ref{lem:oracle_donsker} gives weak convergence of the leading empirical process to $\mathbb G_a$. The distributional treatment effect statement follows by applying the same argument jointly to the vector class $\{(\phi_1^y,\phi_0^{y'}):y,y'\in\cY\}$ and using the continuous mapping theorem for subtraction.
\end{proof}

\begin{proof}[Proof of Corollary \ref{cor:uniform_quantile}]
By Theorem \ref{thm:uniform} and the assumed $o_p(n^{-1/2})$ distance between $\widehat F_a$ and $\widehat\psi_a$,
\[
\sqrt n(\widehat F_a-\psi_a)\dto \mathbb G_a
\quad\text{in }\ell^\infty(\cY).
\]
Assumption \ref{ass:quantile} implies that the inverse-CDF map is Hadamard differentiable uniformly on the compact set $\cT$ at $\psi_a$, tangentially to continuous functions, with derivative
\[
h\mapsto -h\{q_a(\cdot)\}/f_a\{q_a(\cdot)\}.
\]
The functional delta method gives the quantile-process limit. The QTE limit follows by subtracting the two treatment-specific quantile limits. The final statement follows from Proposition \ref{prop:grid_projection}, which supplies the required $o_p(n^{-1/2})$ replacement of the continuum estimator by the growing-grid interpolated and, when permitted, projected estimator.
\end{proof}

\begin{proof}[Proof of Corollary \ref{cor:multiplier}]
This is a high-level estimated-process result. Let $\Phi_\Delta=\{\phi_{\Delta,\tau}:\tau\in\cT\}$ be the oracle QTE influence-function class. Corollary \ref{cor:uniform_quantile} identifies the weak limit of $n^{-1/2}\sum_i\phi_{\Delta,\tau}(O_i)$. By Assumption \ref{ass:uniform_oracle} and the continuity and bounded-away-from-zero density condition, $\Phi_\Delta$ is $P$-Donsker with a square-integrable envelope. The multiplier theorem is applied to the empirically centered versions of these classes; centering by $\Pn$ only adds data-dependent constants and preserves the stated entropy and envelope orders.

The stated conditions on $\widehat\Phi_\Delta$ give two ingredients. First,
\[
\sup_{\tau\in\cT}\|\widehat\phi_{\Delta,\tau}-\phi_{\Delta,\tau}\|_{P,2}=o_p(1),
\]
including the effect of replacing $q_a,f_a,$ and nuisance functions by estimates. Second, conditional on the data, $\widehat\Phi_\Delta$ has a VC-type entropy bound and bounded envelope of the same order as the oracle class. The conditional multiplier maximal inequality applied to the centered difference class therefore yields
\[
\sup_{\tau\in\cT}\left|\frac1{\sqrt n}\sum_{i=1}^n\xi_i\Big[\{\widehat\phi_{\Delta,\tau}(O_i)-\Pn\widehat\phi_{\Delta,\tau}\}-\{\phi_{\Delta,\tau}(O_i)-\Pn\phi_{\Delta,\tau}\}\Big]\right|=o_p(1)
\]
conditionally in probability. The multiplier central limit theorem for the centered oracle Donsker class gives conditional weak convergence of
\[
\frac1{\sqrt n}\sum_{i=1}^n\xi_i\{\phi_{\Delta,\tau}(O_i)-\Pn\phi_{\Delta,\tau}\}
\]
to the same Gaussian limit as the original empirical process. Combining the two displays gives conditional weak convergence of $\mathbb Z_n^*$ in $\ell^\infty(\cT)$. The continuous mapping theorem transfers convergence to the supremum norm, and continuity of the limiting supremum distribution at its $(1-\alpha)$ quantile gives the stated simultaneous band validity.
\end{proof}

\begin{proof}[Proof of Proposition \ref{prop:grid_projection}]
Let $I_n\psi_a$ be the linear interpolant of the true grid values $\psi_a(y_{j,n})$. Since $\psi_a$ is continuously differentiable with bounded derivative, $\sup_{y\in\cY}|I_n\psi_a(y)-\psi_a(y)|=O(h_n)=o(n^{-1/2})$. Let $I_n(\Pn\phi_a)$ be the linear interpolant of the empirical-process grid values. The assumed local empirical modulus implies
\[
\sup_{y\in\cY}|I_n(\Pn\phi_a)(y)-\Pn\phi_a^y|=o_p(n^{-1/2}).
\]
For any $y$ between adjacent grid points, linear interpolation expresses $\widetilde\psi_{a,n}(y)-I_n\psi_a(y)-I_n(\Pn\phi_a)(y)$ as the same convex combination of the two endpoint remainders; hence its supremum is bounded by the maximum grid remainder. Combining these three bounds gives the displayed uniform expansion on $\cY$.

If $\widehat F_{a,n}$ is a nondecreasing version satisfying $\sup_{y\in\cY_q}|\widehat F_{a,n}(y)-\widetilde\psi_{a,n}(y)|=o_p(n^{-1/2})$, the same expansion holds for $\widehat F_{a,n}$ on $\cY_q$. The inverse-map limits then follow from the pointwise and uniform Bahadur arguments used in Corollaries \ref{cor:pointwise} and \ref{cor:uniform_quantile}, since $f_a$ is bounded away from zero on the quantile region.

On a fixed grid with strictly increasing true grid values, let $c=\min_j\{\psi_a(y_{j+1})-\psi_a(y_j)\}>0$. If $\max_j|\widehat\psi_a(y_j)-\psi_a(y_j)|<c/2$, then the unprojected grid estimator is already increasing and the isotonic projection is inactive. On a growing grid, adjacent true increments can be smaller than the stochastic grid error, so monotonicity of the raw interpolant need not follow from the first display alone. For any grid, isotonic regression is Euclidean projection onto a closed convex cone intersected with a box, so it is nonexpansive in Euclidean norm. If the projection distance is $o_p(n^{-1/2})$ in sup norm on the quantile region, adding the projection changes neither the CDF expansion nor the inverse-CDF expansion at first order.
\end{proof}

\section{Proofs for Lower-Level Nuisance-Rate Verification}

\begin{proof}[Proof of Lemma \ref{lem:uniform_mrate}]
For each $y$, write the empirical objective as $\mathbb M_{n,y}(\theta)+\lambda_n\|\theta\|_2^2/2$ and the population objective as $\mathbb M_y(\theta)$. The argument is a finite-dimensional sieve M-estimation calculation of the same type used for nonparametric and series estimators \citep{Stone1980,Stone1985,Newey1997}. The first-order condition and a Taylor expansion around $\theta_y^\star$ give
\[
0=\nabla\mathbb M_{n,y}(\theta_y^\star)+\widehat H_{n,y}(\widehat\theta_y-\theta_y^\star)+\lambda_n\widehat\theta_y,
\]
where $\widehat H_{n,y}$ is an empirical Hessian along the segment from $\theta_y^\star$ to $\widehat\theta_y$. Uniform positive definiteness implies $\lambda_{\min}(\widehat H_{n,y})>c/2$ with probability tending to one, uniformly in $y$, because the empirical Hessians are uniformly consistent. For a fixed $y$, the score vector is an average of bounded mean-zero $d_n$-vectors, so its Euclidean norm is $O_p(\sqrt{d_n/n})$. A union bound over $J_n$ grid points and Bernstein's inequality add the maximal term $O_p(\sqrt{\log(eJ_n)/n})$. Thus
\[
\sup_{y\in\cY_n}\|\nabla\mathbb M_{n,y}(\theta_y^\star)-\nabla\mathbb M_y(\theta_y^\star)\|_2
=O_p\left(\sqrt{\frac{d_n+\log(eJ_n)}{n}}\right).
\]
Since $\nabla\mathbb M_y(\theta_y^\star)=0$ and $\|\theta_y^\star\|_2\le B_n$, the penalty contributes $O(\lambda_n B_n)$. Inverting the Hessian gives the coefficient rate. The uniform upper eigenvalue bound on the population Gram matrix converts coefficient error into $L_2$ fitted-function error with only a constant-factor loss. The sieve approximation error is added by the triangle inequality, and Lipschitz interpolation adds $O(h_n)$ for continuum $\cY$.
\end{proof}

\begin{proof}[Proof of Theorem \ref{thm:sieve_rates}]
Apply Lemma \ref{lem:uniform_mrate} to each nuisance regression. For $m_a$, the relevant dimension is $d_{U,n}$ and the index set is $\cY_n$, giving $r_{m,n}$. Under positivity, the validation subsample size for $R=1,A=a,M=1$ is proportional to $n$ with probability tending to one, so the denominator $n$ may be used in the rate up to constants. For $e_a$, there is no $y$ index and the relevant dimension is $d_{X,n}$, giving $r_{e,n}$. For $\rho_a$, the dimension is $d_{U,n}$ and there is no $y$ index, giving $r_{\rho,n}$.

For $\omega$, the calibrated ridge logistic source-target classifier is a convex M-estimation problem in dimension $d_{X,n}$; the calibration formulation follows the same moment-balancing logic as survey calibration and empirical likelihood \citep{DevilleSarndal1992,QinLawless1994,ImaiRatkovic2014}; the transformation from $\widehat p_0(x)$ to $\widehat\omega(x)$ is Lipschitz on the positivity region after truncation and normalization, so the same finite-dimensional M-estimation argument gives $r_{\omega,n}$ when the classifier model has approximation error $a_{\omega,n}$. Under entropy balancing, Assumption \ref{ass:sieve_regular} restricts the density ratio to the stated exponential-tilt sieve. The convex dual has moment equation
\[
\Pnr{1}\{\omega_\theta(X)b_{X,n}(X)\}-\Pnr{0}b_{X,n}(X)=0
\]
with normalization $\Pnr{1}\omega_\theta=1$. The population moment vanishes at the true or pseudo-true tilt parameter up to approximation error $a_{\omega,n}$, the empirical moment is $O_p(\sqrt{d_{X,n}/n})$, and the Jacobian is the covariance matrix of $b_{X,n}$ under the tilted source law, nonsingular by Assumption \ref{ass:sieve_regular}. Taylor expansion and boundedness of the exponential-tilt map on the coefficient neighborhood give the same rate for the tilt parameter and hence for $\widehat\omega$ in $L_2(P_{X,1})$ when evaluated at new $x$ values.

For $g_a$, nested cross-fitting makes the generated response independent of the second-stage evaluation fold conditional on the first-stage training data. Decompose
\[
\widehat g_a-g_a=(\widehat g_a-g_a^{\widehat m})+(g_a^{\widehat m}-g_a),
\]
where $g_a^{\widehat m}(y,x)=E\{\widehat m_a(y,X,S)\mid R=1,A=a,X=x\}$. The first term is the second-stage ridge regression error and has rate $a_{g,n}+\sqrt{(d_{X,n}+\log(eJ_n))/n}+h_n$. The second term is bounded by the first-stage error because conditional expectation is an $L_2$ contraction:
\[
\sup_y\|g_a^{\widehat m}(y)-g_a(y)\|_{P_{X,1},2}\le \sup_y\|\widehat m_a(y)-m_a(y)\|_{1a,2}=O_p(r_{m,n}).
\]
This proves the stated $r_{g,n}$ rate. Truncation to the positivity interval is nonexpansive up to constants, and the maps $(\omega,e)\mapsto \omega e_a/e$ and $\rho\mapsto\rho/\widehat\rho$ are Lipschitz on the positivity region, so $r_{w,a,n}=O_p(r_{\omega,n}+r_{e,n})$ and $r_{\rho\mathrm{rat},a,n}=O_p(r_{\rho,n})$.
\end{proof}

\begin{proof}[Proof of Corollary \ref{cor:sieve_products}]
Theorem \ref{thm:sieve_rates} gives the nuisance rates entering the drift bound in Proposition \ref{prop:drift}; substituting the displayed product condition gives an $o_p(n^{-1/2})$ drift uniformly in $y$. The one-step signal is a finite algebraic combination of bounded links applied to the sieve fits and the threshold class $\{\ind(Y\le y):y\in\cY_n\}$, followed by interpolation. Its local perturbation class has entropy dimension of order $D_n=d_{X,n}+d_{U,n}+\log(eJ_n)$ and $L_2(P)$ radius of order $r_{\Gamma,n}$. Hence the entropy-growth condition in Assumption \ref{ass:uniform_nuisance} follows from
\[
r_{\Gamma,n}\sqrt{D_n\log n}+D_n\log n/\sqrt n=o(1).
\]
The special cases follow by substituting the corresponding rates. When $D_n$ is bounded this reduces to the usual fixed-entropy DML requirement; when $D_n$ grows, the additional entropy-growth condition is necessary and is not implied by product rates alone.
\end{proof}

\begin{proof}[Proof of Theorem \ref{thm:rkhs_rates}]
For a fixed $y$, standard KRR theory under a source condition and effective-dimension bound \citep{CaponnettoDeVito2007} yields
\[
\|\widehat f_{y,\lambda}-f_y\|_{P,2}
=O_p\left(\lambda^r+\sqrt{\frac{\mathcal N(\lambda)}{N_{\mathrm{eff}}}}\right).
\]
The first term is the Tikhonov regularization bias and the second is the stochastic term. Applying this bound on the grid $\cY_n$ and taking a union bound over $J_n$ grid points gives the additional factor $\sqrt{\log(eJ_n)}$ in the stochastic term. With $\lambda_n\asymp N_{\mathrm{eff}}^{-1/(2r+\alpha)}$ and $\mathcal N(\lambda)\le C\lambda^{-\alpha}$, the bias and stochastic terms balance at order $N_{\mathrm{eff}}^{-r/(2r+\alpha)}\sqrt{\log(eJ_n)}$, up to constants. Since $N_{\mathrm{eff}}/n$ is bounded away from zero with probability tending to one, this is equivalent to the displayed rate. The Lipschitz condition in $y$ and linear interpolation add $h_n$. Clipping probability regressions is nonexpansive in $L_2$ and hence preserves rates for regression-type probability targets.

The product-rate statement follows by substituting the KRR rates for $m$ and $g$ into Proposition \ref{prop:drift} and using the separately verified rates for $(e,\rho,\omega)$. The theorem deliberately does not assert that an arbitrary KRR regression automatically estimates the density ratio $\omega$; a classifier, calibration, or density-ratio loss must separately justify $s_{\omega,n}$. The remaining assertions are algebraic consequences of the product-rate display and the entropy-growth requirement stated in Corollary \ref{cor:sieve_products}.
\end{proof}

\section{Proofs for Efficiency-Gain Results}

\begin{proof}[Proof of Theorem \ref{thm:eff_gain_cdf}]
Suppress $a,y$ and write $\rho(X)=\rho_a^0(X)$ under Assumption \ref{ass:rho_x}. Conditional on $(R=1,A=a,X)$, define
\[
D=m_a(y,X,S)-g_a(y,X),\qquad \varepsilon=Z_y-m_a(y,X,S).
\]
Then $E(D\mid R=1,A=a,X)=0$, $E(\varepsilon\mid R=1,A=a,X,S,M=1)=0$, and $E(D\varepsilon\mid R=1,A=a,X)=0$.

The source component of the surrogate-assisted efficient influence function is
\[
C_S=\frac{\omega(X)\ind(A=a)}{\pi_1 e_a(X)}\left\{D+\frac{M}{\rho(X)}\varepsilon\right\}.
\]
The source component of the no-surrogate benchmark efficient influence function is
\[
C_0=\frac{\omega(X)\ind(A=a)}{\pi_1 e_a(X)}\frac{M}{\rho(X)}(D+\varepsilon).
\]
The target-covariate components of $\phi_a^y$ and $\phi_{a,0}^y$ are identical, and the source components have conditional mean zero given $X$. Hence the variance difference is the integrated difference between the conditional variances of $C_0$ and $C_S$.

Conditional on $(R=1,A=a,X)$,
\begin{align*}
E\left[\left(D+\frac{M}{\rho}\varepsilon\right)^2\middle|X,A=a,R=1\right]
&=\V(D\mid X,A=a,R=1)\\
&\quad+\frac1\rho E\{\V(Z_y\mid X,S,A=a,R=1)\mid X,A=a,R=1\},\\
E\left[\left\{\frac{M}{\rho}(D+\varepsilon)\right\}^2\middle|X,A=a,R=1\right]
&=\frac1\rho\V(D\mid X,A=a,R=1)\\
&\quad+\frac1\rho E\{\V(Z_y\mid X,S,A=a,R=1)\mid X,A=a,R=1\}.
\end{align*}
Their difference is $(\rho^{-1}-1)\V(D\mid X,A=a,R=1)$. Multiplying by the squared outer weight and integrating over $(X,A)$ among source units gives
\[
V_{a,0}(y)-V_a(y)
=\frac1{\pi_1}E_1\left[\frac{\omega^2(X)}{e_a(X)}\frac{1-\rho_a^0(X)}{\rho_a^0(X)}\V\{m_a(y,X,S)\mid R=1,A=a,X\}\right].
\]
Nonnegativity and strictness follow immediately.
\end{proof}

\begin{proof}[Proof of Corollary \ref{cor:eff_gain_quantile}]
By Corollary \ref{cor:eif_quantile}, the efficiency bound for $q_a(\tau)$ is the CDF efficiency bound at $q_a(\tau)$ divided by $f_a^2(q_a(\tau))$. Applying Theorem \ref{thm:eff_gain_cdf} at $y=q_a(\tau)$ proves the first claim. For the QTE, the target-covariate component is unchanged by observing $S$, while the treatment-specific source components for $a=0$ and $a=1$ have disjoint treatment support and conditional mean zero. Thus the variance reduction for the QTE equals the sum of the two treatment-specific quantile variance reductions.
\end{proof}

\section{More Results of Simulation Studies}

\subsection{Experiment 1 with \texorpdfstring{$n=4000$}{n=4000}}\label{app:exp1_n4000}

Table \ref{tab:exp1_appendix_n4000} repeats Experiment 1 with $n=4000$.  The larger sample reduces MSE for the one-step estimators and preserves the same qualitative conclusion as the main table: SA remains the strongest feasible point estimator across all reported quantiles, while IPW remains more variable and Plugin and Source remain biased.

\begin{table}[!htbp]
\centering
\begin{threeparttable}
\caption{Supplementary finite-sample performance for transported QTE estimation, $n=4000$}
\label{tab:exp1_appendix_n4000}
\small
\setlength{\tabcolsep}{5.5pt}
\renewcommand{\arraystretch}{1.12}
\begin{tabular}{cc l r c r r r}
\toprule
$\tau$ & $\Delta_0(\tau)$ & Method & Bias & MSE & RMSE & Coverage & CI length \\
\midrule
0.25 & -3.884 & IPW        &  0.121 & \msegain{1.645}{52.8} & 1.283 & 0.794 & 3.706 \\
0.25 & -3.884 & Oracle     & -0.003 & \mseplain{0.676}      & 0.822 & 0.822 & 2.473 \\
0.25 & -3.884 & Plugin     &  1.914 & \msegain{4.020}{80.7} & 2.005 & \dash & \dash \\
\SArow
0.25 & -3.884 & \textbf{SA} & \textbf{0.141} & \msebest{0.776} & \textbf{0.881} & \textbf{0.800} & \textbf{2.508} \\
0.25 & -3.884 & Source     &  1.761 & \msegain{3.191}{75.7} & 1.786 & \dash & \dash \\
\addlinespace[0.35em]
0.50 & -0.940 & IPW        &  0.119 & \msegain{2.428}{67.4} & 1.558 & 0.822 & 5.144 \\
0.50 & -0.940 & Oracle     & -0.066 & \mseplain{0.505}      & 0.710 & 0.808 & 2.130 \\
0.50 & -0.940 & Plugin     &  2.897 & \msegain{8.736}{90.9} & 2.956 & \dash & \dash \\
\SArow
0.50 & -0.940 & \textbf{SA} & \textbf{0.185} & \msebest{0.792} & \textbf{0.890} & \textbf{0.808} & \textbf{2.647} \\
0.50 & -0.940 & Source     &  2.870 & \msegain{8.335}{90.5} & 2.887 & \dash & \dash \\
\addlinespace[0.35em]
0.75 &  3.150 & IPW        &  0.105 & \msegain{6.503}{83.9} & 2.550 & 0.858 & 7.897 \\
0.75 &  3.150 & Oracle     & -0.028 & \mseplain{0.416}      & 0.645 & 0.814 & 1.709 \\
0.75 &  3.150 & Plugin     &  2.785 & \msegain{7.992}{86.9} & 2.827 & \dash & \dash \\
\SArow
0.75 &  3.150 & \textbf{SA} & \textbf{0.241} & \msebest{1.045} & \textbf{1.022} & \textbf{0.782} & \textbf{2.650} \\
0.75 &  3.150 & Source     &  2.827 & \msegain{8.077}{87.1} & 2.842 & \dash & \dash \\
\bottomrule
\end{tabular}
\begin{tablenotes}[flushleft]
\footnotesize
\item \textit{Notes.} See Table \ref{tab:exp1_main_n2000}. The superscript next to MSE again reports the percentage MSE reduction achieved by SA relative to the corresponding feasible baseline.
\end{tablenotes}
\end{threeparttable}
\end{table}

\FloatBarrier
\subsection{Experiment 2 efficiency diagnostics}\label{app:exp2_efficiency}

This subsection provides diagnostic summaries for Experiment 2 beyond the efficiency ratios reported in Table~\ref{tab:exp2_efficiency_tau50} and Appendix Table~\ref{tab:exp2_appendix_ratios}. 
The ratio tables establish the main efficiency pattern: the advantage of SA over NoS increases with surrogate predictiveness and is largest when the primary outcome is sparsely validated. 
The figures below address a different question: whether those gains are accompanied by acceptable finite-sample calibration and whether they reflect lower dispersion rather than systematic bias. 
Each point in the boxplots corresponds to one design cell, indexed by a quantile level, a validation rate, and a surrogate strength.

Figures~\ref{fig:exp2_a3_coverage_length} and~\ref{fig:exp2_a3_bias_mcsd} complement the efficiency-ratio results rather than duplicating them.  Oracle uses the true nuisance functions under the same observed-data structure as SA.
Full oracle additionally treats all source primary outcomes as observed and is included
only as a full-data lower-bound benchmark.

The coverage plot shows that SA has reasonable finite-sample calibration across the design grid; the slight undercoverage in the most difficult low-validation settings is also visible for the oracle-type one-step estimators and is consistent with finite-sample quantile inversion and density estimation effects. 
The interval-length and MC-SD ratios show the main practical implication of using the surrogate: relative to SA, NoS and IPW require substantially wider intervals and exhibit larger empirical dispersion. 
At the same time, the bias plot indicates that the gain is not produced by shifting the estimator toward the truth in a design-specific way. 
Thus, the diagnostics support the interpretation of Experiment 2 as an efficiency gain from observing and using the surrogate process, rather than as an artifact of bias, model misspecification, or overly narrow standard errors.

\begin{figure}
    \centering
    \includegraphics[width=1\linewidth]{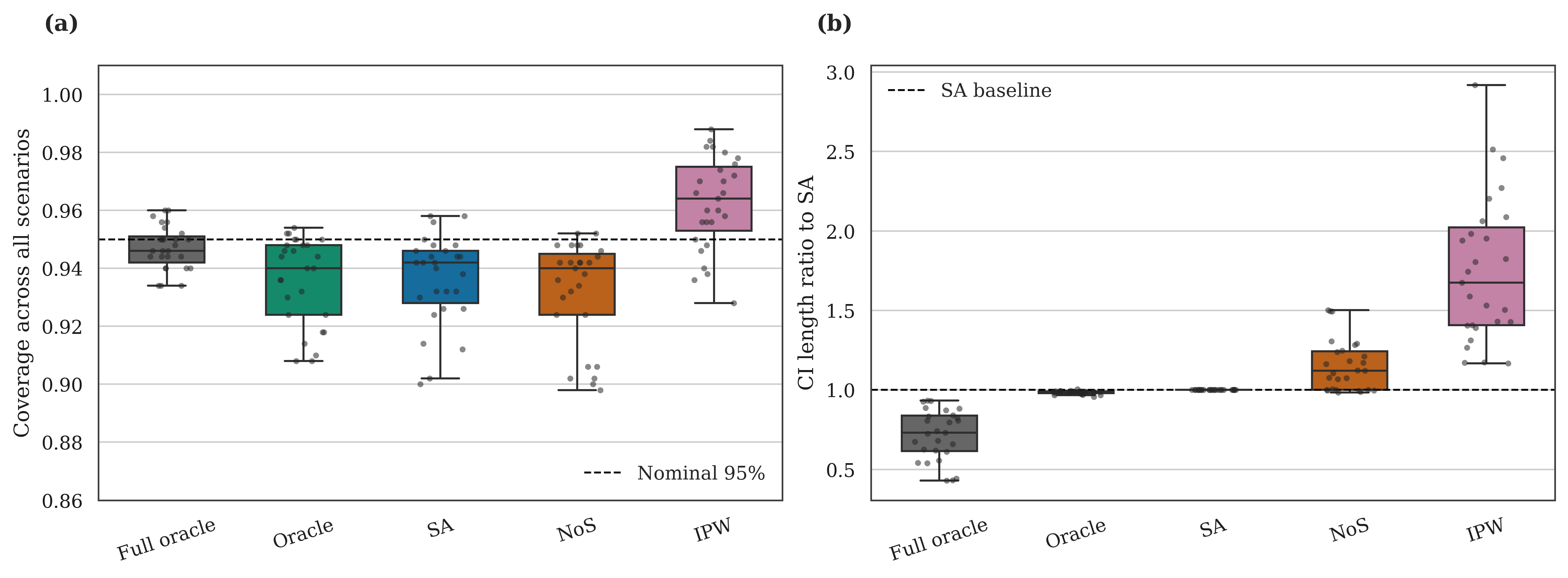}
    \caption{Experiment 2 finite-sample diagnostics: coverage calibration and interval efficiency. 
    Panel (a) reports empirical coverage of nominal 95\% Wald intervals across all Experiment 2 design cells. 
    The horizontal dashed line marks the nominal 95\% level. 
    Panel (b) reports the ratio of average confidence interval length to that of SA in the same design cell. 
    Values above one indicate intervals longer than SA. 
    SA attains reasonable finite-sample coverage while producing shorter intervals than the no-surrogate and IPW estimators; IPW is more conservative largely because its intervals are substantially longer.}
    \label{fig:exp2_a3_coverage_length}
\end{figure}

\begin{figure}
    \centering
    \includegraphics[width=1\linewidth]{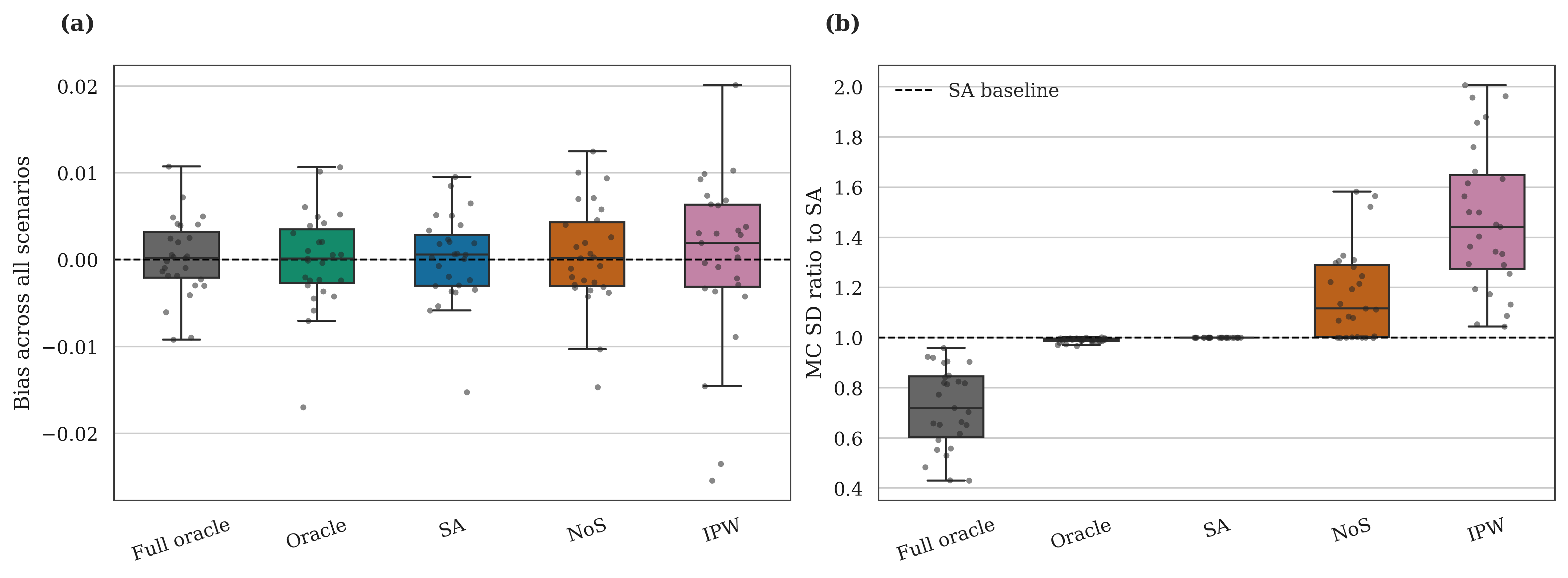}
    \caption{Experiment 2 finite-sample diagnostics: bias and empirical dispersion. 
    Panel (a) reports Monte Carlo bias across all Experiment 2 design cells. 
    Biases are centered near zero for all one-step estimators, indicating that the MSE gains in Table~\ref{tab:exp2_efficiency_tau50} and Appendix Table~\ref{tab:exp2_appendix_ratios} are not driven by systematic bias. 
    Panel (b) reports the ratio of Monte Carlo standard deviation to that of SA in the same design cell. 
    Values above one indicate larger empirical dispersion than SA. 
    The no-surrogate and IPW estimators are systematically more variable than SA, while the infeasible oracle benchmarks provide lower-bound references.}
    \label{fig:exp2_a3_bias_mcsd}
\end{figure}

\begin{table}[!htbp]
\centering
\begin{adjustbox}{max width=\linewidth}
\begin{threeparttable}
\caption{Experiment 2 appendix: efficiency ratios across all quantiles}
\label{tab:exp2_appendix_ratios}

\scriptsize
\setlength{\tabcolsep}{3.0pt}
\renewcommand{\arraystretch}{1.08}

\begin{tabular}{C{0.055\linewidth}C{0.070\linewidth}C{0.075\linewidth}R{0.085\linewidth}R{0.085\linewidth}R{0.095\linewidth}R{0.075\linewidth}R{0.105\linewidth}R{0.105\linewidth}}
\toprule
$\tau$ & $\bar\rho$ & $\lambda_S$ & MSE(SA) & MSE(NoS) & NoS/SA MSE & SA gain & NoS/SA length & Theory NoS/SA \\
\midrule
0.25 & 0.20 & 0.00 & 0.012 & 0.012 & 1.01 & 0.7\% & 0.99 & 1.00 \\
0.25 & 0.20 & 1.00 & 0.023 & 0.040 & 1.71 & 41.5\% & 1.21 & 1.52 \\
0.25 & 0.20 & 2.00 & 0.050 & 0.125 & 2.50 & 60.0\% & 1.49 & 2.26 \\
0.25 & 0.40 & 0.00 & 0.006 & 0.006 & 1.00 & 0.2\% & 1.00 & 1.00 \\
0.25 & 0.40 & 1.00 & 0.012 & 0.017 & 1.43 & 29.9\% & 1.18 & 1.36 \\
0.25 & 0.40 & 2.00 & 0.034 & 0.053 & 1.55 & 35.5\% & 1.29 & 1.68 \\
0.25 & 0.70 & 0.00 & 0.003 & 0.003 & 1.00 & -0.0\% & 1.00 & 1.00 \\
0.25 & 0.70 & 1.00 & 0.009 & 0.011 & 1.24 & 19.3\% & 1.08 & 1.15 \\
0.25 & 0.70 & 2.00 & 0.027 & 0.032 & 1.18 & 15.1\% & 1.12 & 1.25 \\
0.50 & 0.20 & 0.00 & 0.010 & 0.010 & 1.00 & 0.2\% & 0.98 & 1.00 \\
0.50 & 0.20 & 1.00 & 0.023 & 0.037 & 1.64 & 39.2\% & 1.25 & 1.59 \\
0.50 & 0.20 & 2.00 & 0.044 & 0.101 & 2.32 & 56.8\% & 1.50 & 2.19 \\
0.50 & 0.40 & 0.00 & 0.005 & 0.005 & 1.00 & 0.1\% & 1.00 & 1.00 \\
0.50 & 0.40 & 1.00 & 0.013 & 0.019 & 1.48 & 32.3\% & 1.17 & 1.39 \\
0.50 & 0.40 & 2.00 & 0.029 & 0.050 & 1.68 & 40.6\% & 1.31 & 1.71 \\
0.50 & 0.70 & 0.00 & 0.003 & 0.003 & 1.00 & -0.1\% & 1.00 & 1.00 \\
0.50 & 0.70 & 1.00 & 0.008 & 0.009 & 1.16 & 14.1\% & 1.07 & 1.15 \\
0.50 & 0.70 & 2.00 & 0.020 & 0.025 & 1.29 & 22.5\% & 1.12 & 1.27 \\
0.75 & 0.20 & 0.00 & 0.012 & 0.012 & 1.01 & 1.0\% & 1.00 & 1.00 \\
0.75 & 0.20 & 1.00 & 0.026 & 0.044 & 1.72 & 41.8\% & 1.24 & 1.57 \\
0.75 & 0.20 & 2.00 & 0.052 & 0.129 & 2.45 & 59.2\% & 1.50 & 2.24 \\
0.75 & 0.40 & 0.00 & 0.006 & 0.006 & 1.00 & -0.2\% & 0.99 & 1.00 \\
0.75 & 0.40 & 1.00 & 0.014 & 0.021 & 1.49 & 32.8\% & 1.16 & 1.38 \\
0.75 & 0.40 & 2.00 & 0.032 & 0.057 & 1.76 & 43.3\% & 1.28 & 1.68 \\
0.75 & 0.70 & 0.00 & 0.004 & 0.004 & 1.00 & -0.1\% & 1.00 & 1.00 \\
0.75 & 0.70 & 1.00 & 0.010 & 0.012 & 1.15 & 12.8\% & 1.07 & 1.15 \\
0.75 & 0.70 & 2.00 & 0.028 & 0.034 & 1.25 & 19.7\% & 1.11 & 1.26 \\
\bottomrule
\end{tabular}

\begin{tablenotes}[flushleft]
\footnotesize
\item \textit{Notes.} This table repeats the main efficiency comparison at all reported quantiles. Ratios larger than one favor SA over the no-surrogate benchmark.
\end{tablenotes}

\end{threeparttable}
\end{adjustbox}
\end{table}

\FloatBarrier
\subsection{Experiment 3 source-target covariate shift and density-ratio adjustment}\label{app:exp3_omega}

Appendix Figure~\ref{fig:exp3_omega_diagnostics} separates two sources of difficulty in Experiment 3. 
The first is the intrinsic overlap problem induced by stronger source-target covariate shift; the second is the statistical error from estimating the density ratio. As expected, Panel (a) shows that the fitted density ratio becomes less accurate as \(c\) increases. However, Panel (b) shows that replacing the true \(\omega\) by \(\widehat\omega\) has little effect on the MSE of the proposed SA estimator: the estimated-\(\omega\)/true-\(\omega\) MSE ratio remains close to one under both randomized and observational assignment. This supports the interpretation of the main Experiment 3 results: the source-population estimator fails because it targets the wrong covariate law, whereas the transported SA estimator remains stable; the modest deterioration at larger \(c\) reflects reduced overlap rather than a breakdown of the density-ratio estimator.

\begin{figure}[!t]
    \centering
    \includegraphics[width=0.96\linewidth]{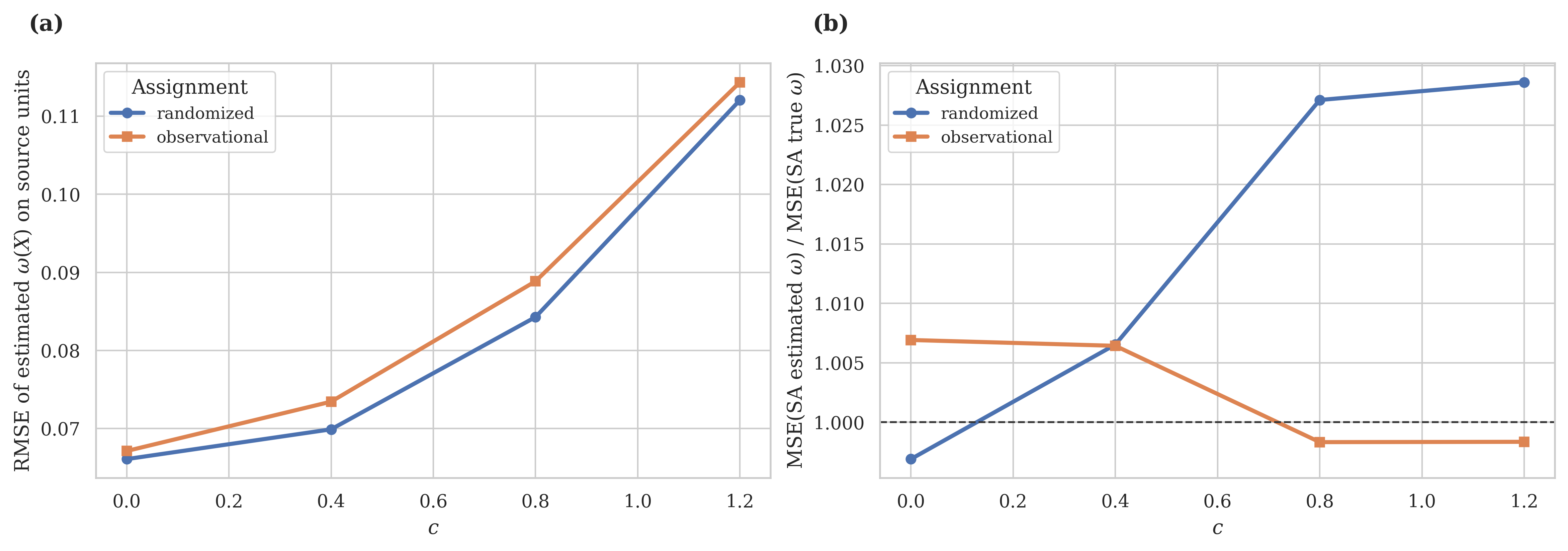}
    \caption{
    Experiment 3 density-ratio diagnostics under source-target covariate shift.
    Panel (a) reports the root mean squared error of the fitted density ratio \(\widehat\omega(X)\) on source units as the transport-shift parameter \(c\) increases.
    The increasing trend confirms that density-ratio estimation becomes harder as the target and source covariate laws move farther apart.
    Panel (b) reports the MSE ratio of the feasible SA estimator using estimated \(\widehat\omega\) to the same SA estimator using the true density ratio \(\omega\).
    Ratios close to one indicate that, over the shift range considered, the feasible density-ratio learner does not materially inflate the MSE of the transported SA estimator.
    Thus the performance degradation observed under larger \(c\) is better interpreted as finite-sample overlap loss rather than density-ratio estimation failure.
    }
    \label{fig:exp3_omega_diagnostics}
\end{figure}

\FloatBarrier
\subsection{Experiment 4 simultaneous band width and isotonic-projection diagnostics}\label{app:exp4_simultaneous}

Appendix Figure~\ref{fig:exp4_width_iso} provides the implementation diagnostics associated with the uniform-band results in Table~\ref{tab:exp4_uniform_grid}. Panel (a) shows the expected decrease in average simultaneous-band width as the sample size grows. The growing grid produces wider bands in finite samples, which is a conservative numerical consequence of evaluating the CDF and quantile process on a finer grid rather than evidence against the growing-grid condition. Panel (b) examines the isotonic projection step. Although the raw one-step CDF estimates are projected to enforce monotonicity before inversion, the scaled distance \(\sqrt n d_{\rm iso}\) decreases with \(n\) for all grid choices. This supports the interpretation that isotonic projection stabilizes the finite-sample CDF estimates without becoming a first-order source of error in the QTE expansion.

\begin{figure}[!t]
    \centering
    \includegraphics[width=0.96\linewidth]{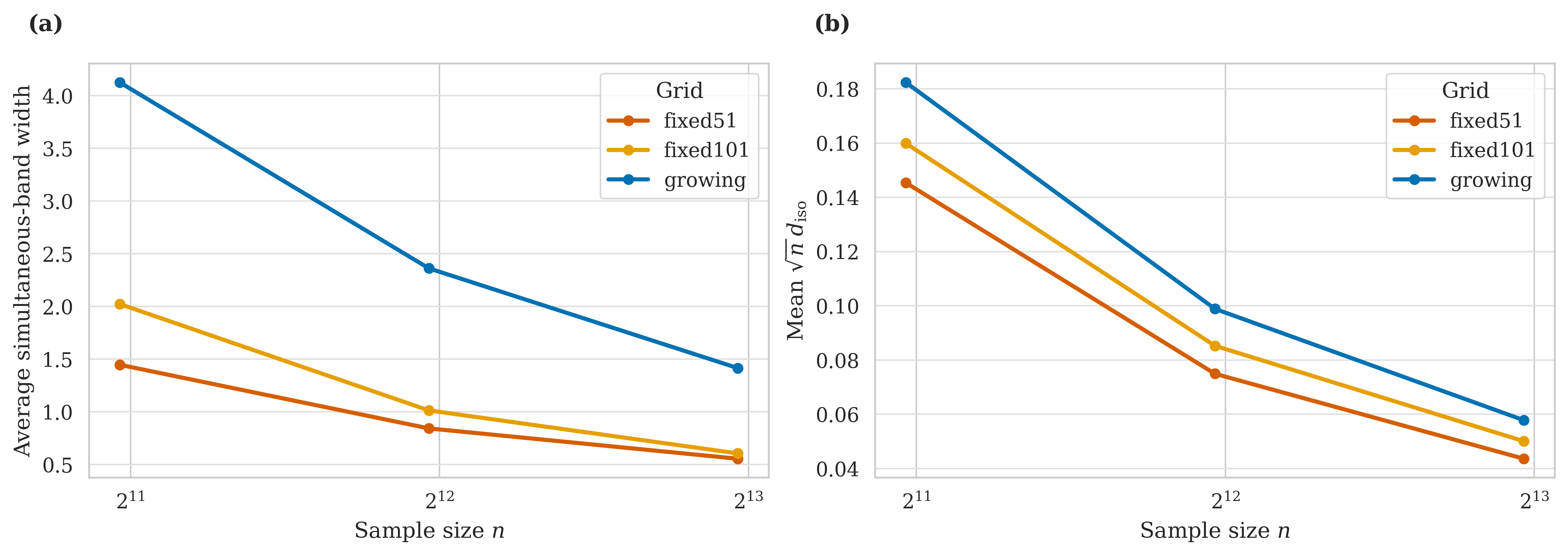}
    \caption{
    Experiment 4 appendix: simultaneous band width and isotonic-projection diagnostics. 
    Panel (a) reports the average width of the simultaneous QTE bands over \(\mathcal T=\{0.10,0.11,\ldots,0.90\}\). 
    Band widths decrease with sample size for all grid choices, while the growing grid is more conservative in finite samples because it resolves the CDF and quantile process on a substantially finer numerical grid. 
    Panel (b) reports the mean scaled isotonic-projection distance \(\sqrt n d_{\rm iso}\), where \(d_{\rm iso}\) is the maximum distance between the raw one-step CDF estimate and its isotonic projection over treatment arms and grid points. 
    The decreasing trend indicates that the monotonicity correction does not dominate the first-order quantile expansion.
    }
    \label{fig:exp4_width_iso}
\end{figure}

\FloatBarrier
\section{ACTG 175 implementation details}\label{app:actg175_details}

This appendix provides implementation details for the ACTG 175 empirical illustration in Section~\ref{sec:realdata}. 
The analysis uses the public ACTG 175 data set \citep{hammer1996trial,actg175data}. 
The primary outcome is 96-week CD4 count, \(Y=\texttt{cd496}\), and \(M=\texttt{r}\) indicates whether \(Y\) is observed. 
The surrogate vector is \(S=(\texttt{cd420},\texttt{cd820})\), the 20-week CD4 and CD8 counts. 
The baseline covariate vector \(X\) contains age, weight, hemophilia status, homosexual activity, history of intravenous drug use, Karnofsky score, prior antiretroviral therapy indicators, antiretroviral history, race, gender, symptomatic status, and baseline CD4/CD8 counts. 
The treatment indicator \(A\) is the binary variable \texttt{treat}, coded as non-zidovudine-only therapy versus zidovudine-only therapy.

\paragraph{Internal transport construction.}
Because ACTG 175 is a randomized trial and does not include a separate external target-only sample, we construct an internal target covariate distribution using baseline variables only. 
Specifically, a fixed baseline-only score is formed from standardized baseline CD4 count, previous antiretroviral therapy duration, age, symptomatic status, and antiretroviral treatment history. 
The intercept is calibrated so that approximately 35\% of subjects are assigned to the target sample. 
Let \(R=0\) denote the constructed target sample and \(R=1\) the source sample. 
The split is stochastic rather than deterministic, preserving overlap between the source and target covariate distributions. 
After the split is fixed, the estimators use only \(X\) from target units. 
Target treatment, surrogate, validation, and outcome variables are not used in estimation.

\paragraph{Estimators.}
The feasible estimators are the same as in the simulation studies. 
SA is the proposed transported surrogate-assisted one-step estimator. 
NoS is the transported no-surrogate estimator that targets the same transported QTE but does not use \(S\). 
IPW is a validation-only transported weighting estimator. 
Source is a negative-control estimator of the source-population QTE and is included to show how the transported target estimand differs from the source-population estimand. 
Oracle estimators are not reported in this real-data analysis because the nuisance functions and the target QTE truth are unknown. 
Likewise, the source-population negative control is reported only as a point-estimation comparison.

\paragraph{Design diagnostics.}
Table~\ref{tab:actg175_design} summarizes the constructed analysis sample. 
The source validation rate is 0.627, so the real data contain a nontrivial missing-primary-outcome component. 
The density-ratio effective sample size is 682.9, corresponding to 48.5\% of the source sample. 
The comparison of predictive \(R^2\) values shows that adding the 20-week CD4/CD8 surrogate vector to the baseline covariates increases the prediction of 96-week CD4 in the source validation sample from 0.313 to 0.504. 
This supports the empirical relevance of the surrogate-assisted component.

\begin{table}[!htbp]
\centering
\begin{threeparttable}
\caption{ACTG 175 implementation diagnostics}
\label{tab:actg175_design}
\small
\renewcommand{\arraystretch}{1.05}
\begin{tabular*}{\textwidth}{@{\extracolsep{\fill}}lrrrr}
\toprule
Diagnostic & Value \\
\midrule
Total sample size & 2139 \\
Source sample size & 1408 \\
Target sample size & 731 \\
Source validation rate & 0.627 \\
Source treatment rate & 0.737 \\
Target fraction & 0.342 \\
Density-ratio ESS & 682.9 \\
Density-ratio ESS / source size & 0.485 \\
\(R^2\), baseline-only prediction of 96-week CD4 & 0.313 \\
\(R^2\), baseline plus 20-week CD4/CD8 prediction & 0.504 \\
Incremental \(R^2\) from surrogate & 0.191 \\
\bottomrule
\end{tabular*}
\begin{tablenotes}[flushleft]
\footnotesize
\item \textit{Notes.}
The \(R^2\) diagnostics are computed in the source validation sample, where 96-week CD4 is observed. 
ESS denotes the effective source sample size induced by the fitted density-ratio weights. 
The target sample contributes only baseline covariates to the transported analysis.
\end{tablenotes}
\end{threeparttable}
\end{table}

\begin{figure}
    \centering
    \includegraphics[width=0.9\linewidth]{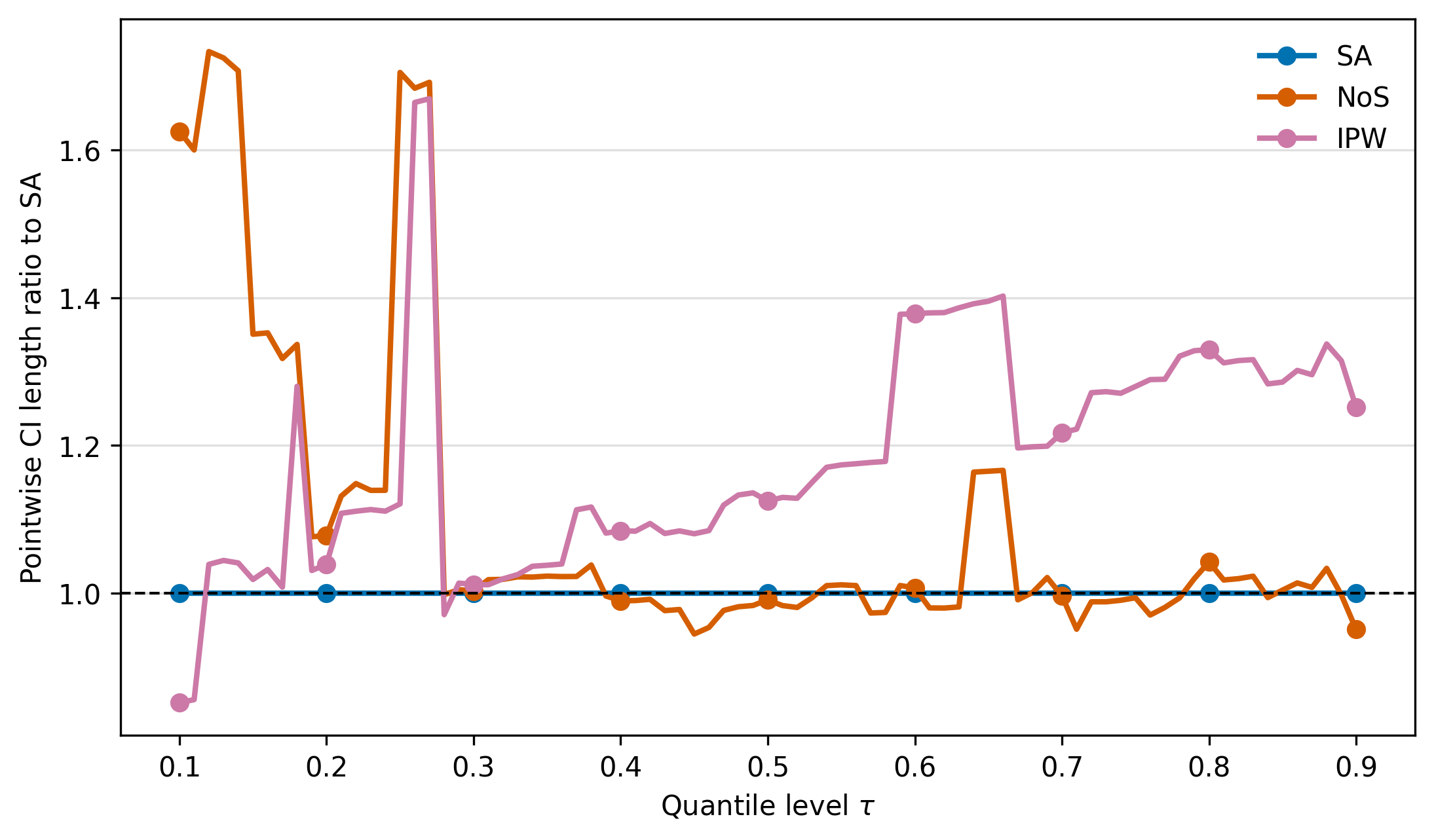}
    \caption{
    ACTG 175 precision diagnostics. 
    The figure reports the pointwise confidence-interval length ratio relative to SA across quantile levels. 
    The SA ratio is normalized to one. 
    Values above one indicate wider intervals than SA. 
    The validation-only IPW estimator is generally less stable, especially in the upper half of the distribution. 
    NoS is comparable to SA over much of the distribution but is substantially less precise at lower quantiles, consistent with the point estimates in Table~\ref{tab:actg175_qte}.
    }
    \label{fig:actg175_precision_ratios}
\end{figure}

\begin{figure}
    \centering
    \includegraphics[width=0.9\linewidth]{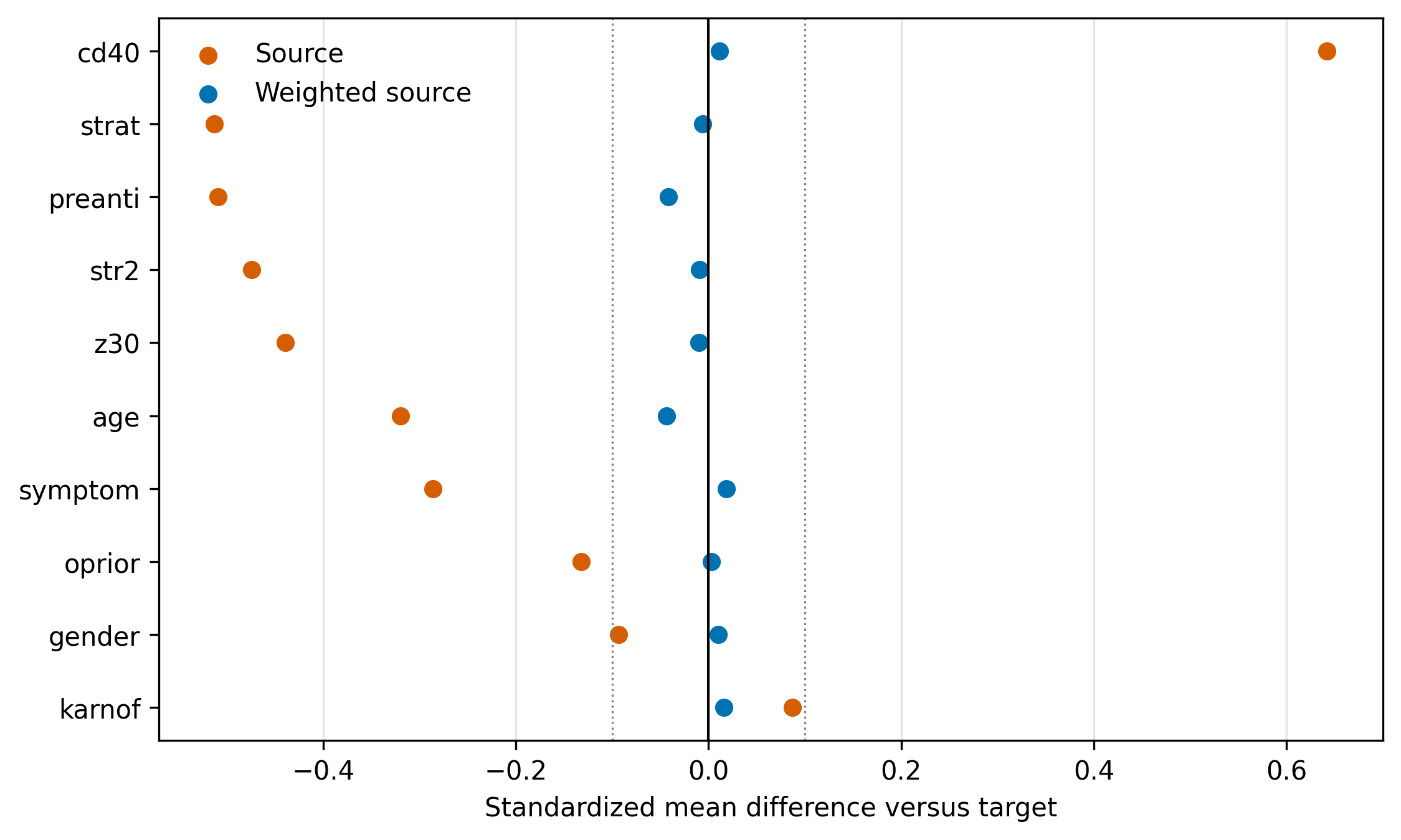}
    \caption{
    ACTG 175 source-target balance diagnostics. 
    The figure reports standardized mean differences between the source and target samples before and after density-ratio weighting. 
    The unweighted source sample differs substantially from the constructed target covariate distribution on several baseline variables, including baseline CD4 and antiretroviral treatment-history variables. 
    Weighting by \(\widehat\omega(X)\) reduces these imbalances, supporting the interpretation of the analysis as a transported target-population QTE rather than a source-population comparison.
    }
    \label{fig:actg175_transport_balance}
\end{figure}

Figures~\ref{fig:actg175_precision_ratios} and~\ref{fig:actg175_transport_balance} provide additional diagnostics for interpreting the empirical illustration. 
The precision-ratio plot shows where the surrogate-assisted estimator yields shorter intervals than the no-surrogate and validation-only alternatives. 
The balance plot shows that the density-ratio component materially changes the source covariate distribution toward the target distribution. 
Together, these diagnostics support the use of the full transported surrogate-assisted estimator in this application: the surrogate is predictive, the primary outcome is incompletely observed, and the source and target covariate laws are not interchangeable.

\clearpage
\bibliographystyle{apalike}

\end{document}